\documentclass[sts]{imsart}

\RequirePackage[OT1]{fontenc}
\usepackage{amsthm,amsmath,natbib,amsfonts,enumerate,bbm,graphics,graphicx}
\RequirePackage[colorlinks,citecolor=blue,urlcolor=blue]{hyperref}
\usepackage[ruled, vlined]{algorithm2e}

\startlocaldefs
\numberwithin{equation}{section}
\theoremstyle{plain}
\newtheorem{thm}{Theorem}[section]
\newtheorem{prop}[thm]{Proposition}
\newtheorem{lemma}[thm]{Lemma}
\newtheorem{corollary}[thm]{Corollary}

\DeclareMathOperator*{\argmax}{argmax}

\endlocaldefs

\begin{document}

\begin{frontmatter}
\title{Recent progress in log-concave density estimation}
\runtitle{Log-concave density estimation}

\begin{aug}
\author{\fnms{Richard} J. \snm{Samworth}\thanksref{t1}\ead[label=e1]{r.samworth@statslab.cam.ac.uk}}
\thankstext{t1}{Supported by an EPSRC Early Career Fellowship, an EPSRC Programme grant and a grant from the Leverhulme Trust.}
\runauthor{R. J. Samworth}

\affiliation{University of Cambridge}

\address{Statistical Laboratory, Wilberforce Road, Cambridge CB3 0WB,United Kingdom \printead{e1}.}

\end{aug}

\begin{abstract}
In recent years, log-concave density estimation via maximum likelihood estimation has emerged as a fascinating alternative to traditional nonparametric smoothing techniques, such as kernel density estimation, which require the choice of one or more bandwidths.  The purpose of this article is to describe some of the properties of the class of log-concave densities on $\mathbb{R}^d$ which make it so attractive from a statistical perspective, and to outline the latest methodological, theoretical and computational advances in the area.
\end{abstract}

\begin{keyword}
\kwd{Log-concavity}
\kwd{maximum likelihood estimation}
\end{keyword}

\end{frontmatter}

\section{Introduction}

Shape-constrained density estimation has a long history, dating back at least as far as \citet{Grenander1956}, who studied the maximum likelihood estimator of a decreasing density on the non-negative half-line.  Unlike traditional nonparametric smoothing approaches, this estimator does not require the choice of any tuning parameter, and indeed it has a beautiful characterisation as the left derivative of the least concave majorant of the empirical distribution function.  Over subsequent years, a great deal of work went into understanding its theoretical properties \citep[e.g.][]{PrakasaRao1969,Groeneboom1985,Birge1989}, revealing in particular its non-standard cube-root rate of convergence.

On the other hand, the class of decreasing densities on $[0,\infty)$ is quite restrictive, and does not generalise particularly naturally to multivariate settings.  In recent years, therefore, alternative families of densities have been sought, and the class of log-concave densities has emerged as one with many attractive properties from a statistical viewpoint.  This has led to applications of the theory to a wide variety of problems, including the detection of the presence of mixing \citep{Walther2002}, filtering \citep{HenningssonAstrom2006}, tail index estimation \citep{MullerRufibach2009}, clustering \citep{CSS2010}, regression \citep{DSS2011}, Independent Component Analysis \citep{SamworthYuan2012} and classification \citep{ChenSamworth2013}.  

The main aim of this article is to give an account of the key properties of log-concave densities and their relevance for applications in statistical problems.  We focus especially on ideas of log-concave projection, which underpin the maximum likelihood approach to inference within the class.  Recent theoretical results and computational aspects will also be discussed.  For alternative reviews of related topics, see \citet{SaumardWellner2014}, which has a greater emphasis on analytic properties, and \citet{Walther2009}, with a stronger focus on modelling and applications.

\section{Basic properties}

We say that $f:\mathbb{R}^d \rightarrow [0,\infty)$ is log-concave if $\log f$ is a concave function (with the convention $\log 0 := -\infty$).  Let $\mathcal{F}_d$ denote the class of upper semi-continuous log-concave densities on $\mathbb{R}^d$ with respect to $d$-dimensional Lebesgue measure.  The upper semi-continuity is not particularly important in most of what follows, but it fixes a particular version of the density and means we do not need to worry about densities that differ on a set of zero Lebesgue measure.  We do not consider here degenerate log-concave densities whose support is contained in a lower-dimensional affine subset of $\mathbb{R}^d$.

Many standard families of densities are log-concave.  For instance, Gaussian densities with positive-definite covariance matrices and uniform densities on convex, compact sets belong to $\mathcal{F}_d$; the logistic density $f(x) = \frac{e^{-x}}{(1-e^{-x})^2}$, Beta$(a,b)$ densities with $a,b, \geq 1$, Weibull$(\alpha)$ denities with $\alpha \geq 1$, $\Gamma(\alpha,\lambda)$ densities with $\alpha \geq 1$, Gumbel and Laplace densities (amongst many others) belong to $\mathcal{F}_1$.  It is convenient to think of log-concave densities as unimodal densities with exponentially decaying tails.  Unimodality here is meant in the sense of the upper level sets being convex, though in one dimension, we have a stronger characterisation: 
\begin{lemma}[\citet{Ibragimov1956}]
A density $f$ on $\mathbb{R}$ is log-concave if and only if the convolution $f \ast g$ is unimodal for every unimodal density $g$. 
\end{lemma}
A more precise statement about the exponentially decaying tails is as follows:
\begin{lemma}[\citet{CuleSamworth2010}]
If $f \in \mathcal{F}_d$, then there exist $\alpha > 0$, $\beta \in \mathbb{R}$ such that $f(x) \leq e^{-\alpha\|x\| + \beta}$ for all $x \in \mathbb{R}^d$.
\end{lemma}
Thus, in particular, random vectors with log-concave densities have moment generating functions that are finite in a neighbourhood of the origin.

One of the features of the class of log-concave densities that makes them so attractive for statistical inference is their stability under various operations.  A key result of this type is the following, due to \citet{Prekopa1973}, and with a simpler proof given in \citet{Prekopa1980}.
\begin{thm}
\label{Thm:Marg}
Let $d = d_1 + d_2$ for some $d_1,d_2 \in \mathbb{N}$, and let $f:\mathbb{R}^d \rightarrow [0,\infty)$ be log-concave.  Then 
\[
x \mapsto \int_{\mathbb{R}^{d_2}} f(x,y) \, dy
\]
is log-concave on $\mathbb{R}^{d_1}$.
\end{thm}
Hence, marginal densities of log-concave random vectors are log-concave.  As a simple consequence, we have
\begin{corollary}
\label{Cor:Conv}
If $f,g$ are log-concave densities on $\mathbb{R}^d$, then their convolution $f \ast g$ is a log-concave density on $\mathbb{R}^d$.  
\end{corollary} 
\begin{proof}
The function $(x,y) \mapsto f(x-y)g(y)$ is log-concave on $\mathbb{R}^{2d}$, so the result follows from Theorem~\ref{Thm:Marg}.
\end{proof}
Two further straightforward stability properties are as follows:
\begin{prop}
\label{Prop:Lin}
Let $X$ have a log-concave density $f$ on $\mathbb{R}^d$.
\begin{enumerate}[(i)]
\item If $A \in \mathbb{R}^{m \times d}$ has $m \leq d$ and $\mathrm{rank}(A) = m$, then $AX$ has a log-concave density on $\mathbb{R}^m$.
\item If $X = (X_1^\top,X_2^\top)^\top$, then the conditional density of $X_1$ given $X_2=x_2$ is log-concave for each $x_2$.
\end{enumerate}
\end{prop}
Together, Theorem~\ref{Thm:Marg}, Corollary~\ref{Cor:Conv} and Proposition~\ref{Prop:Lin} indicate that the class of log-concave densities is a natural infinite-dimensional generalisation of the class of Gaussian densities.  Indeed, one can argue that a grand vision in the shape-constrained inference community is to free practitioners from restrictive parametric (often Gaussian) assumptions, while retaining many of the properties of these parametric procedures that make them so convenient for use in applications.

\section{Log-concave projections}
\label{Sec:LCProjections}

Despite all of the nice properties of $\mathcal{F}_d$ described in the previous section, the class is not convex.  It is therefore by no means clear that there should exist a `closest' element of this set to a general distribution.  Nevertheless, it turns out that one can make sense of such a notion, and that the appropriate concept is that of log-concave projection.

Let $\Phi$ denote the class of upper semi-continuous, concave functions $\phi:\mathbb{R}^d \rightarrow [-\infty,\infty)$ that are coercive in the sense that $\phi(x) \rightarrow -\infty$ as $\|x\| \rightarrow \infty$.  Thus $\mathcal{F}_d = \bigl\{e^\phi:\phi \in \Phi, \int_{\mathbb{R}^d} e^\phi = 1\bigr\}$.  For $\phi \in \Phi$ and an arbitrary probability measure $P$ on $\mathbb{R}^d$, define a kind of log-likelihood functional by
\[
L(\phi,P) := \int_{\mathbb{R}^d} \phi \, dP - \int_{\mathbb{R}^d} e^\phi.
\]
Thus, instead of enforcing the (non-convex) constraint that $\phi$ should be a log-density explicitly, the functional above has the flavour of a Lagrangian, though the Lagrange multiplier is conspicuous by its absence!  Nevertheless it turns out that any maximiser $\phi^* \in \Phi$ of this functional with $L(\phi^*,P) \in \mathbb{R}$ must be a log-density.  To see this, note that if $\phi \in \Phi$ has $L(\phi,P) \in \mathbb{R}$ and $c \in \mathbb{R}$, then
\[
\frac{\partial}{\partial c} L(\phi+c,P) = 1 - e^c \int_{\mathbb{R}^d} e^\phi.
\]
Hence, at a maximum, $c = -\log\bigl(\int_{\mathbb{R}^d} e^\phi\bigr)$, which is equivalent to $\phi + c$ being a log-density. 

Theorem~\ref{Thm:ExistProj} below gives a complete characterisation of when there exists a unique maximiser of $L(\phi,P)$ over $\phi \in \Phi$.  We first require several further definitions: let $L^*(P) := \sup_{\phi \in \Phi} L(\phi,P)$ and let $\mathcal{P}_d$ denote the class of probability measures $P$ on $\mathbb{R}^d$ satisfying both $\int_{\mathbb{R}^d} \|x\| \, dP(x) < \infty$ and $P(H) < 1$ for all hyperplanes $H$.  Let $\mathcal{C}_d$ denote the class of closed, convex subsets of $\mathbb{R}^d$, for a probability measure $P$ on $\mathbb{R}^d$, let $\mathcal{C}_d(P) := \{C \in \mathcal{C}_d:P(C) = 1\}$, and let $\mathrm{csupp}(P) := \cap_{C \in \mathcal{C}_d(P)} C$ denote the convex support of $P$.  Finally, let $\mathrm{int}(C)$ denote the interior of a convex set $C$, and for a concave function $\phi:\mathbb{R}^d \rightarrow [-\infty,\infty)$, let $\mathrm{dom}(\phi) := \{x:\phi(x) > -\infty\}$ denote its effective domain.
\begin{thm}[\citet{DSS2011}] \ \\ \vspace{-0.6cm}
\label{Thm:ExistProj} 
\begin{enumerate}[(i)]
\item If $\int_{\mathbb{R}^d} \|x\| \, dP(x) = \infty$, then $L^*(P) = -\infty$.
\item If $\int_{\mathbb{R}^d} \|x\| \, dP(x) < \infty$ but $P(H) = 1$ for some hyperplane $H$, then $L^*(P) = \infty$.
\item If $P \in \mathcal{P}_d$, then $L^*(P) \in \mathbb{R}$ and there exists a unique $\phi^* \in \Phi$ that maximises $L(\phi,P)$ over $\phi \in \Phi$.  Moreover, $\mathrm{int}\bigl(\mathrm{csupp}(P)\bigr) \subseteq \mathrm{dom}(\phi^*) \subseteq \mathrm{csupp}(P)$.
\end{enumerate}
\end{thm}
A consequence of Theorem~\ref{Thm:ExistProj} and the preceding discussion is that there exists a well-defined map $\psi^*:\mathcal{P}_d \rightarrow \mathcal{F}_d$, given by
\[
\psi^*(P) := \argmax_{f \in \mathcal{F}_d} \int_{\mathbb{R}^d} \log f \, dP.
\]
We refer to $\psi^*$ as the \emph{log-concave projection}.  In the case where $P$ is the empirical distribution of some data, this tells us that provided the convex hull of the data is $d$-dimensional, there exists a unique log-concave maximum likelihood estimator (MLE), a result first proved in \citet{Walther2002} in the case $d=1$, and \citet{CSS2010} for general $d$.  If $P$ has a log-concave density $f_0$, then $\psi^*(P) = f_0$; more generally, if $P$ has a density $f_0$ satisfying $\int_{\mathbb{R}^d} f_0 |\log f_0| < \infty$, then $\psi^*(P)$ minimises the Kullback--Leibler divergence $d_{\mathrm{KL}}^2(f_0,f) := \int_{\mathbb{R}^d} f_0 \log (f_0/f)$ over all $f \in \mathcal{F}_d$.  These statements justify the use of the term `projection'.  

\section{Computation of log-concave maximum likelihood estimators}
\label{Sec:Algorithm}

Let $X_1,\ldots,X_n \stackrel{\mathrm{iid}}{\sim} P \in \mathcal{P}_d$, and let $\mathbb{P}_n$ denote their empirical distribution.  In this section, we discuss the computation of the log-concave MLE $\hat{f}_n := \psi^*(\mathbb{P}_n)$ when the convex hull $C_n$ of $X_1,\ldots,X_n$ is $d$-dimensional.  

We initially focus on the case $d=1$, and follow the Active Set approach of \citet{DHR2007}, which is implemented in the \texttt{R} package \texttt{logcondens} \citep{DumbgenRufibach2011}.  Write $X_{(1)} \leq \ldots \leq X_{(n)}$ for the order statistics of the sample, and let $\Psi$ denote the set of functions $\psi:\mathbb{R} \rightarrow [-\infty,\infty)$ that are continuous on $[X_{(1)},X_{(n)}]$, linear on each $[X_{(k)},X_{(k+1)}]$ and $-\infty$ on $\mathbb{R} \setminus [X_{(1)},X_{(n)}]$.  Let $\Psi_{\mathrm{conc}}$ denote the concave functions in $\Psi$.  Then $\log \hat{f}_n \in \Psi_{\mathrm{conc}}$, because otherwise we could strictly increase $L(\cdot,\mathbb{P}_n)$ by replacing $\log \hat{f}_n$ with the $\psi \in \Psi_{\mathrm{conc}}$ with $\psi(X_i) = \log \hat{f}_n(X_i)$.  Since any $\psi \in \Psi$ can be identified with the vector $\underline{\psi} := \bigl(\psi(X_{(1)}),\ldots,\psi(X_{(n)})\bigr)^\top \in \mathbb{R}^n$, our objective function can be written as
\[
\tilde{L}(\underline{\psi}) = \tilde{L}(\psi_1,\ldots,\psi_n) := \frac{1}{n}\sum_{i=1}^n \psi_i - \sum_{k=1}^{n-1} \delta_k J(\psi_k,\psi_{k+1}),
\]
where $\delta_k := X_{(k+1)} - X_{(k)}$ (assumed positive for simplicity) and 
\[
J(r,s) := \int_0^1 e^{(1-t)r + ts} \, dt.
\]

For $j=2,\ldots,n-1$, let $v_j = (v_{j,1},\ldots,v_{j,n})^\top \in \mathbb{R}^n$ have three non-zero components:
\[
v_{j,j-1} := \frac{1}{\delta_{j-1}}, \quad v_{j,j} := -\frac{1}{\delta_j} - \frac{1}{\delta_{j-1}}, \quad v_{j,j+1} := \frac{1}{\delta_j}.
\]
Then our optimisation problem can be expressed as:
\[
\text{Maximise } \tilde{L}(\underline{\psi}) \quad \text{over } \underline{\psi} \in \mathcal{K} := \bigl\{\underline{\psi} \in \mathbb{R}^n:v_j^\top \underline{\psi} \leq 0 \text{ for } j=2,\ldots,n-1\bigr\}.
\]
For any $\underline{\psi} \in \mathbb{R}^n$, we can define the set of `active' constraints $A(\underline{\psi}) := \bigl\{j \in \{2,\ldots,n-1\}:v_j^\top \underline{\psi} \geq 0\bigr\}$, so that for $\underline{\psi} \in \mathcal{K}$, the inactive constraints correspond to the `knots' of $\underline{\psi}$, where $\underline{\psi}$ changes slope.  Since $\tilde{L}$ is strictly concave and infinitely differentiable, for any $A \subseteq \{2,\ldots,n-1\}$ and corresponding subspace $\mathcal{V}(A) := \bigl\{\underline{\psi} \in \mathbb{R}^n:v_j^\top\underline{\psi} = 0 \text{ for } j \in A\bigr\}$, it is straightforward to compute
\[
\tilde{\psi}(A) \in \mathcal{V}_*(A) := \argmax_{\underline{\psi} \in \mathcal{V}(A)} \tilde{L}(\underline{\psi}).
\] 
using Newton methods.  The basic idea of the Active Set approach is to start at a feasible point with a given active set of variables $A$.  We then optimise the objective under that set of active constraints, and move there if that new candidate point is feasible.  If not, we move as far as we can along the line segment joining our current feasible point to the candidate point while remaining feasible.  This new point has a strictly larger active set than our previous iterate, so we can optimise the objective under this new set of active constraints, and repeat.  More precisely, define a basis for $\mathbb{R}^n$ by $b_1 := (1)_{i=1}^n$, $b_j := \min(X_{(i)} - X_{(j)},0)_{i=1}^n$ for $j=2,\ldots,n-1$ and $b_n := (X_{(i)})_{i=1}^n$.  By considering the first-order stationarity conditions, it can be shown that any $\underline{\psi} \in \mathcal{V}_*(A)$ maximises $\tilde{L}$ over $\mathcal{K}$ if and only if $b_j^\top \nabla \tilde{L}(\underline{\psi}) \leq 0$ for all $j \in A$.  The Active Set algorithm can therefore proceed as in Algorithm~\ref{Alg:ActiveSet}.
\begin{algorithm}[htbp!]
\label{Alg:ActiveSet}
\SetAlgoLined
\IncMargin{1em}
\DontPrintSemicolon
\SetKwRepeat{Do}{while}{end while}
\KwIn{
$A \leftarrow \{2,\ldots,n-1\}$
}
{$\underline{\psi} \leftarrow \tilde{\psi}(A)$}

\While{$\max_{j \in A} b_j^\top\nabla\tilde{L}(\underline{\psi}) > 0$}{
$j^* \leftarrow \min\bigl(\argmax_{j \in A} b_j^\top \nabla \tilde{L}(\underline{\psi})\bigr)$ \\
$\underline{\psi}_{\mathrm{cand}} \leftarrow \tilde{\psi}(A \setminus \{j^*\})$ \\

\While{$\underline{\psi}_{\mathrm{cand}} \notin \mathcal{K}$}{$t^* \leftarrow \max\bigl\{t \in [0,1]:(1-t)\underline{\psi} + t \underline{\psi}_{\mathrm{cand}} \in \mathcal{K}\bigr\}$ \\
$\underline{\psi} \leftarrow (1-t^*)\underline{\psi} + t^* \underline{\psi}_{\mathrm{cand}}$ \\
$A \leftarrow A(\underline{\psi})$ \\
$\underline{\psi}_{\mathrm{cand}} \leftarrow \tilde{\psi}(A)$}

$\underline{\psi} \leftarrow \underline{\psi}_{\mathrm{cand}}$ \\
$A \leftarrow A(\underline{\psi})$}
\vskip 0.5ex
\KwOut{$\underline{\psi}$}
\vskip 1ex
\caption{Pseudo-code for an Active Set algorithm to compute $\bigl(\log \hat{f}_n(X_{(1)}),\ldots,\log \hat{f}_n(X_{(n)})\bigr)^\top$.}
\label{Algo:ADMM}
\end{algorithm}

The main points to note in this algorithm are that in each iteration of the inner \textbf{while} loop, the active set increases strictly (which ensures this loop terminates eventually), and that after each iteration of the outer \textbf{while} loop, the log-likelihood has strictly increased, and the current iterate $\underline{\psi}$ belongs to $\mathcal{K} \cap \mathcal{V}_*(A)$ for some $A \subseteq \{2,\ldots,n-1\}$.  It follows that, up to machine precision, the algorithm terminates with the exact solution in finitely many steps.  See Figure~\ref{Fig:1D}.  
\begin{figure}
\centering
\includegraphics[width=0.48\textwidth]{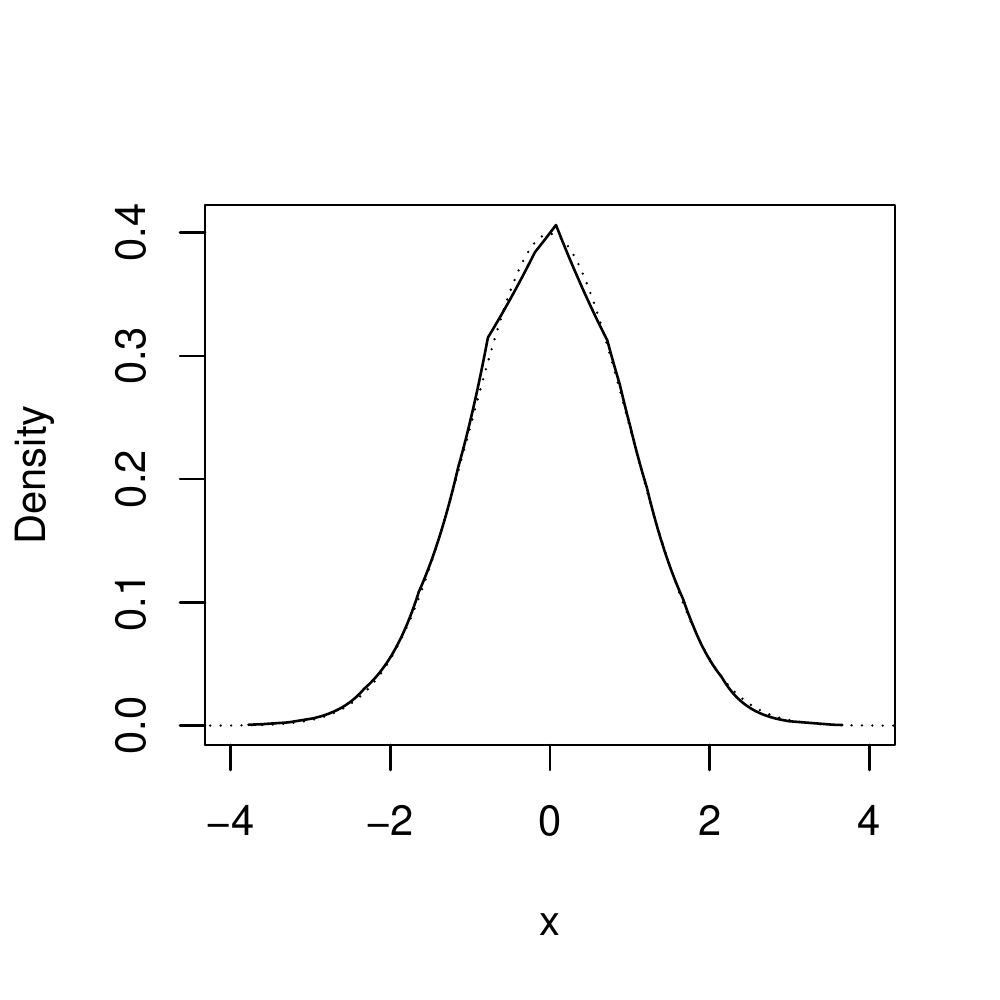} \hspace{0.2cm}
\includegraphics[width=0.48\textwidth]{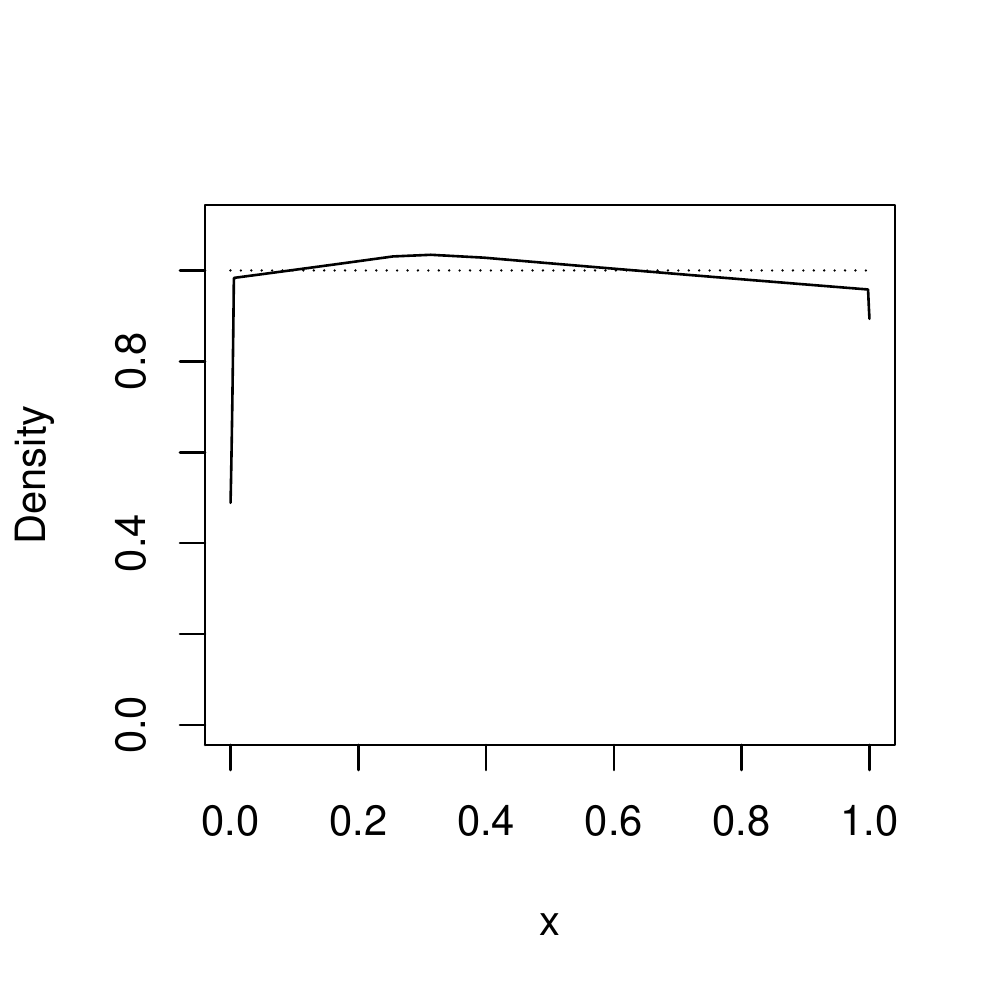}
\caption{\label{Fig:1D}Log-concave maximum likelihood estimators (solid) based on 4000 observations from a standard normal distribution (left) and the $U[0,1]$ distribution (right).  The true densities are shown as dotted lines.}
\end{figure}

For $d \geq 2$, the feasible set is much more complicated, and only slower algorithms are available.  For $y = (y_1,\ldots,y_n) \in \mathbb{R}^n$, let $\bar{h}_y:\mathbb{R}^d \rightarrow \mathbb{R}$ denote the smallest concave function with $\bar{h}_y(X_i) \geq y_i$ for $i=1,\ldots,n$; these are called \emph{tent functions} in \citet{CSS2010} (see Figure~\ref{Fig:Schematic}, which is taken from that paper).
\begin{figure}
\centering
\includegraphics[width=0.6\textwidth]{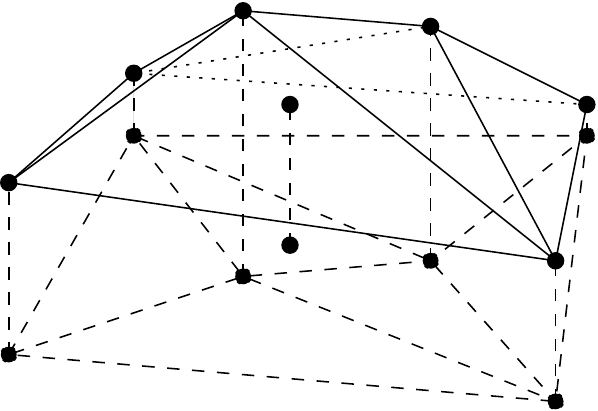}
\caption{\label{Fig:Schematic}A schematic picture of a tent function in the case $d=2$.}
\end{figure}
We can write the objective function in terms of the tent pole heights $y_1,\ldots,y_n$ as 
\[
\tau(y_1,\ldots,y_n) := \frac{1}{n}\sum_{i=1}^n \bar{h}_y(X_i) - \int_{C_n} \exp\{\bar{h}_y(x)\} \, dx.
\]
This function is hard to optimise over $(y_1,\ldots,y_n)^\top \in \mathbb{R}^n$, partly because $\tau$ is not injective.  However, \citet{CSS2010} defined the modified objective function
\[
\sigma(y_1,\ldots,y_n) := \frac{1}{n}\sum_{i=1}^n y_i - \int_{C_n} \exp\{\bar{h}_y(x)\} \, dx.
\]
Thus $\sigma \leq \tau$, but the crucial points are that $\sigma$ is concave and its unique maximum $\hat{y} \in \mathbb{R}^n$ satisfies $\log \hat{f}_n = \bar{h}_{\hat{y}}$.  Even though $\sigma$ is non-differentiable, a subgradient of $-\sigma$ can be computed at every point, so Shor's $r$-algorithm \citep{KappelKuntsevich2000} can be used, as implemented in the \texttt{R} package \texttt{LogConcDEAD} \citep{CGS2009}.  See Figure~\ref{Fig:2D}, which is taken from \citet{CSS2010}.  \citet{KoenkerMizera2010} study an alternative approximate approach based on imposing concavity of the discrete Hessian matrix of the log-density on a grid, and using a Riemann approximation to the integrability constraint.

\begin{figure}
\centering
\includegraphics[width=0.48\textwidth]{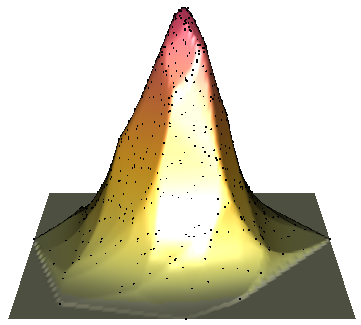} \hspace{0.2cm}
\includegraphics[width=0.48\textwidth]{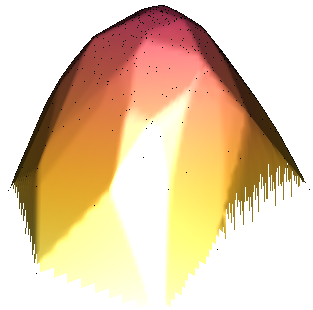}
\caption{\label{Fig:2D}The log-concave maximum likelihood estimator (left) and its logarithm (right) based on 1000 observations from a standard bivariate normal distribution.}
\end{figure}

\section{Properties of log-concave projections}
\label{Sec:Props}

For general distributions $P \in \mathcal{P}_d$, it is not possible to compute the log-concave projection $\psi^*(P)$ explicitly (though see Section~\ref{Sec:1D} below for several exceptions to this).  Nevertheless, one can say quite a lot about the properties of log-concave projections, starting with affine equivariance: 
\begin{lemma}[\citet{DSS2011}]
\label{Lemma:Affine}
Let $X \sim P \in \mathcal{P}_d$, let $A \in \mathbb{R}^{d \times d}$ be invertible, let $b \in \mathbb{R}^d$, and let $P_{A,b}$ denote the distribution of $AX + b$.  Then
\[
\psi^*(P_{A,b})(x) = \frac{1}{|\det A|} \psi^*(P)\bigl(A^{-1}(x-b)\bigr).
\]
\end{lemma}
A generic hope for the log-concave projection is that it should preserve as many properties of the original distribution as possible.  Indeed, as we will see, such preservation results have motivated several associated methodological developments.
\begin{lemma}[\citet{DSS2011}]
\label{Lemma:Basic}
Let $P \in \mathcal{P}_d$, let $\phi^* := \log \psi^*(P)$, and let $P^*(B) := \int_B e^{\phi^*}$ for any Borel set $B \subseteq \mathbb{R}^d$.  If $\Delta:\mathbb{R}^d \rightarrow [-\infty,\infty)$ is such that $\psi^* + t\Delta \in \Phi$ for sufficiently small $t > 0$, then 
\[
\int_{\mathbb{R}^d} \Delta \, dP \leq \int_{\mathbb{R}^d} \Delta \, dP^*.
\]
\end{lemma}
As a special case of Lemma~\ref{Lemma:Basic}, we obtain 
\begin{corollary}
\label{Cor:Moment}
Let $P \in \mathcal{P}_d$.  Then $P$ and the log-concave projection measure $P^*$ from Lemma~\ref{Lemma:Basic} are convex ordered in the sense that 
\[
\int_{\mathbb{R}^d} h \, dP^* \leq \int_{\mathbb{R}^d} h \, dP
\]
for all convex $h:\mathbb{R}^d \rightarrow (-\infty,\infty]$.  
\end{corollary}
Applying Corollary~\ref{Cor:Moment} to $\Delta(x) = t^\top x$ for arbitrary $t \in \mathbb{R}^d$ allows us to conclude that $\int_{\mathbb{R}^d} x \, dP^*(x) = \int_{\mathbb{R}^d} x \, dP(x)$; in other words, log-concave projection preserves the mean $\mu$ of a distribution $P \in \mathcal{P}_d$.  On the other hand, we see that the projection shrinks the second moment, in the sense that $A := \int_{\mathbb{R}^d} (x-\mu)(x-\mu)^\top d(P-P^*)(x)$ is non-negative definite.  This property validates the definition of the \emph{smoothed log-concave projection}, proposed in the case $d=1$ by \citet{DumbgenRufibach2009} and studied for general $d$ in \citet{ChenSamworth2013}.  Writing $\tilde{\mathcal{P}}_d := \bigl\{P \in \mathcal{P}_d:\int_{\mathbb{R}^d} \|x\|^2 \, dP(x) < \infty\bigr\}$, this smoothed projection $\tilde{\psi}^*:\tilde{\mathcal{P}}_d \rightarrow \mathcal{F}_d$ is given by
\[
\tilde{\psi}^*(P) := \psi^*(P) \ast N_d(0,A) = \int_{\mathbb{R}^d} \psi^*(x-y) \, dN_d(0,A)(y).
\]
When $P$ is the empirical distribution of some data, $\tilde{\psi}^*(P)$ is a smooth (real analytic), fully automatic density estimator that is log-concave (cf.~Corollary~\ref{Cor:Conv}), matches the first two moments of the data and is supported on the whole of $\mathbb{R}^d$.  See Figure~\ref{Fig:Smoothed}.
\begin{figure}
\centering
\includegraphics[width=0.48\textwidth]{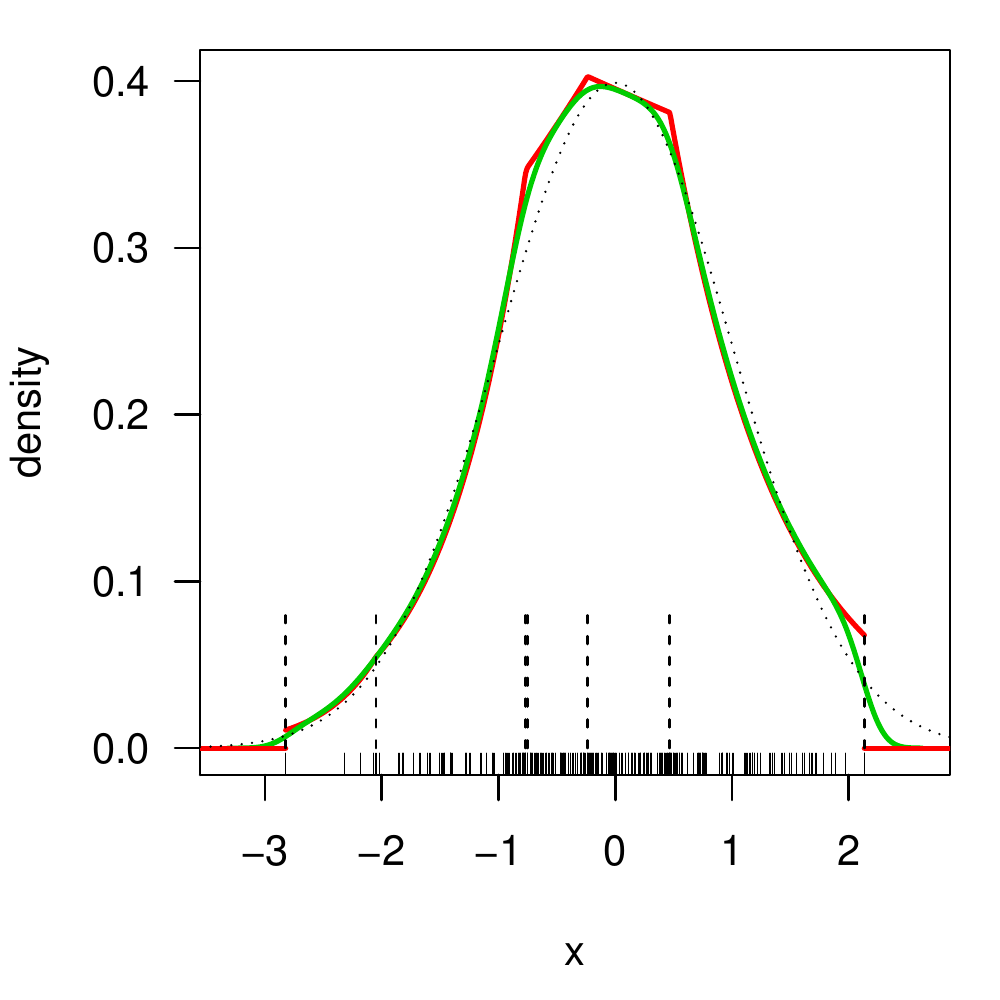} \hspace{0.2cm}
\includegraphics[width=0.48\textwidth]{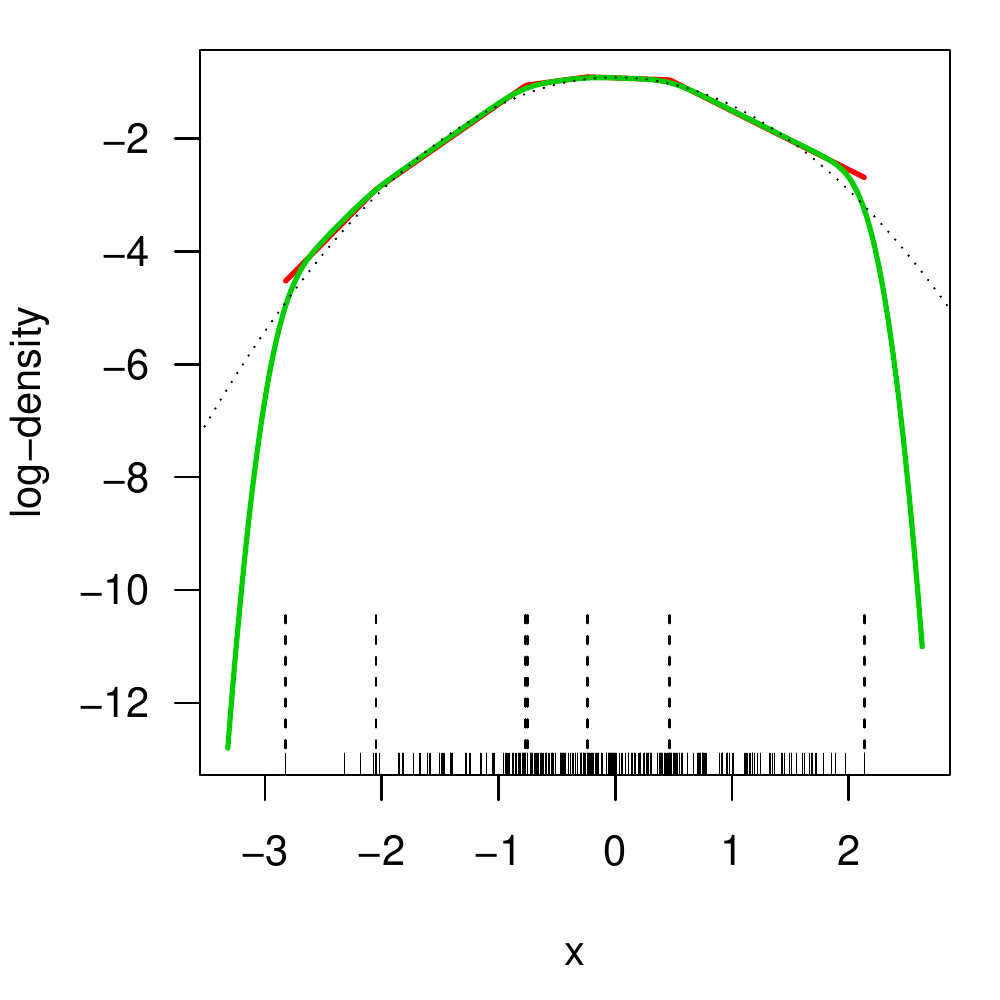}
\caption{\label{Fig:Smoothed}Left: A comparison of the original log-concave MLE (red) and smoothed log-concave MLE (green) based on 200 observations from a standard normal density (dotted).  The short vertical lines indicate the observations, and the longer, dashed vertical lines show the locations of the knots of the log-concave MLE. Right: The same comparison on the log scale.}
\end{figure}

Our next property concerns the preservation of product structure, or, in the language of random vectors, independence of components.
\begin{prop}[\citet{ChenSamworth2013}]
\label{Prop:Indep}
Let $P \in \mathcal{P}_d$ be of the form $P = P_1 \otimes P_2$ for some $P_1 \in \mathcal{P}_{d_1}$, $P_2 \in \mathcal{P}_{d_2}$ with $d_1 + d_2 = d$.  Then for $x = (x_1^\top,x_2^\top)^\top$, we have
\[
\psi^*(P)(x) = \psi^*(P_1)(x_1)\psi^*(P_2)(x_2).
\]
\end{prop}
Proposition~\ref{Prop:Indep} inspires a new approach to Independent Component Analysis; see Section~\ref{Sec:ICA} below.  Incidentally, the converse of this result is false: for instance, for $q \in (0,1]$, consider a distribution $P$ supported on five points in $\mathbb{R}^2$, with
\begin{align*}
P\bigl(\{(0,0)\}\bigr) &= q, \\
P\bigl(\{(-1,-1)\}\bigr) &= P\bigl(\{(-1,1)\}\bigr) = P\bigl(\{(1,-1)\}\bigr) = P\bigl(\{(1,1)\}\bigr) = (1-q)/4.
\end{align*}
Then it can be shown that $\psi^*(P)$ is the uniform density on the square $[-1,1] \times [-1,1]$ for $q \in (0,1/3]$.

In a similar spirit, it is not necessarily the case that the log-concave projection of a marginal distribution is the corresponding marginal of a joint distribution.  For example, if $P$ is the discrete uniform distribution on the three points $\{(-1,-1), (0,3^{1/2}-1), (1,-1)\}$ in $\mathbb{R}^2$ (which form an equilateral triangle), then the log-concave projection is the continuous uniform density on the triangle, with corresponding marginal density $f_1(x_1) = (1-|x|)\mathbbm{1}_{\{|x| \leq 1\}}$ on the $x$-axis.  On the other hand, the log-concave projection of the discrete uniform distribution on $\{-1,0,1\}$ is the uniform density on $[-1,1]$.  

We conclude this section by mentioning a further property that is not preserved by log-concave projection, namely stochastic ordering.  More precisely, let $P$ and $Q$ be distributions on the real line with\footnote{I thank Min Xu and Yining Chen for helpful conversations leading to this example.} $P(\{0\}) = P(\{1\}) = 1/2$ and $Q(\{0\}) = 1/2$, $Q(\{1\}) = 2/5$, $Q(\{2\}) = 1/10$.  Then $P$ is stochastically smaller than $Q$, in the sense that the respective distribution functions $F$ and $G$ satisfy $F(x) \geq G(x)$ with strict inequality for some $x_0$.  Now $\psi^*(P)$ is the uniform density on $[0,1]$, while it can be shown using the ideas in Section~\ref{Sec:1D} below that $\psi^*(Q)(x) = e^{bx-\beta}$ for $x \in [0,2]$, where $b \in [-1.337,-1.336]$ is the unique real solution to 
\[
\frac{1}{b} - \frac{2}{e^{2b}-1} = \frac{7}{5},
\]
and where $\beta = \log\bigl(\frac{e^{2b}-1}{b}\bigr) \in [-0.3619,-0.3612]$.  In particular, $\psi^*(Q)(0) = e^{-\beta} \geq 1.4 > 1 = \psi^*(P)(0)$, so $\psi^*(P)$ is not stochastically smaller than $\psi^*(Q)$; see Figure~\ref{Fig:StochDom}.
\begin{figure}
\centering
\includegraphics[width=0.7\textwidth]{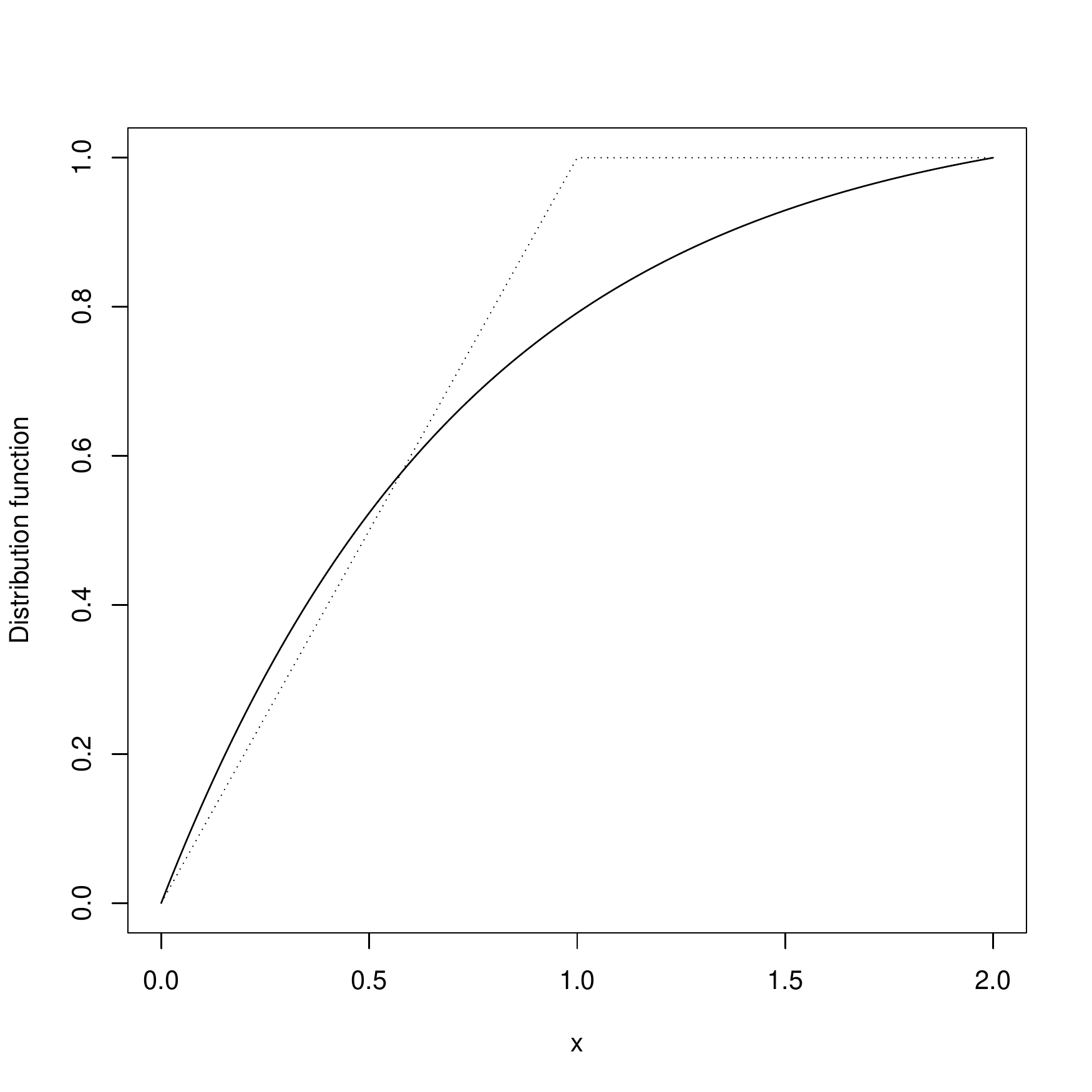}
\caption{\label{Fig:StochDom}The distribution functions corresponding to $\psi^*(P)$ (dotted) and $\psi^*(Q)$ (solid) in the stochastic ordering example at the end of Section~\ref{Sec:Props}.}
\end{figure}      

\section{The one-dimensional case}
\label{Sec:1D}

When $d=1$, the log-concave projection can be characterised in terms of its integrated distribution function.  For $\phi \in \Phi$, let
\[
\mathcal{S}(\phi) := \Bigl\{x \in \mathrm{dom}(\phi):\phi(x) > \frac{1}{2}\{\phi(x+\delta) + \phi(x-\delta)\} \ \text{for all} \ \delta > 0\Bigr\}
\]
denote the closed subset of $\mathbb{R}$ consisting of the points $x_0$ where $\phi$ is not affine in a neighbourhood of $x_0$.  
\begin{thm}[\citet{DSS2011}]
\label{Thm:1D}
Let $P \in \mathcal{P}_1$ have distribution function $F$, and let $F^*$ be a distribution function with density $f^* = e^{\phi^*} \in \mathcal{F}_1$.  Then $f^* = \psi^*(P)$ if and only if
\[
\int_{-\infty}^x \{F^*(t) - F(t)\} \, dt \left\{ \begin{array}{ll} \leq 0 & \mbox{for all $x \in \mathbb{R}$} \\
= 0 & \mbox{for all $x \in \mathcal{S}(\phi^*) \cup \{\infty\}$.} \end{array} \right. 
\]
\end{thm}
In particular, if $P$ is absolutely continuous with respect to Lebesgue measure with continuous density $f$, and if $\mathcal{S}\bigl(\log \psi^*(P)\bigr)$ contains an open interval $I$, then $\psi^*(P) = f$ on $I$.  Theorem~\ref{Thm:1D} is especially useful as a way of verifying the form of log-concave projection in cases where one can guess what it might be.  For instance, consider the family of symmetrised Pareto densities
\[
f(x;\alpha,\sigma) := \frac{\alpha\sigma^\alpha}{2(|x| + \sigma)^{\alpha+1}}, \quad x \in \mathbb{R}, \alpha > 1, \sigma > 0.
\]
Theorem~\ref{Thm:1D} can be used to verify that the corresponding log-concave projection is
\[
f^*(x;\alpha,\sigma) = \frac{\alpha-1}{2\sigma} \exp\biggl\{-\frac{(\alpha-1)|x|}{\sigma}\biggr\}, \quad x \in \mathbb{R};
\]
see \citet{ChenSamworth2013}.  Since the preimage under $\psi^*$ of any $f \in \mathcal{F}_d$ is a convex set, this shows that the preimage of the Laplace density $x \mapsto e^{-|x|}/2$ is infinite-dimensional.  Theorem~\ref{Thm:1D} can also be used to show results such as the following:
\begin{prop}[\citet{DSS2011}]
\label{Prop:ConcConv}
Suppose that $P \in \mathcal{P}_1$ has log-density $\phi$ that is differentiable, convex on a bounded interval $[a,b]$ and concave on $(-\infty,a] \cup [b,\infty)$.  Then there exist $a' \in (-\infty,a]$ and $b' \in [b,\infty)$ such that $\log \psi^*(P)$ is affine on $[a',b']$ and $\log \psi^*(P) = \phi$ on $(-\infty,a'] \cup [b',\infty)$.   
\end{prop}
These ideas are illustrated in Figure~\ref{Fig:Projs}, taken from \citet{DSS2011}.
\begin{figure}
\centering
\includegraphics[width=0.41\textwidth]{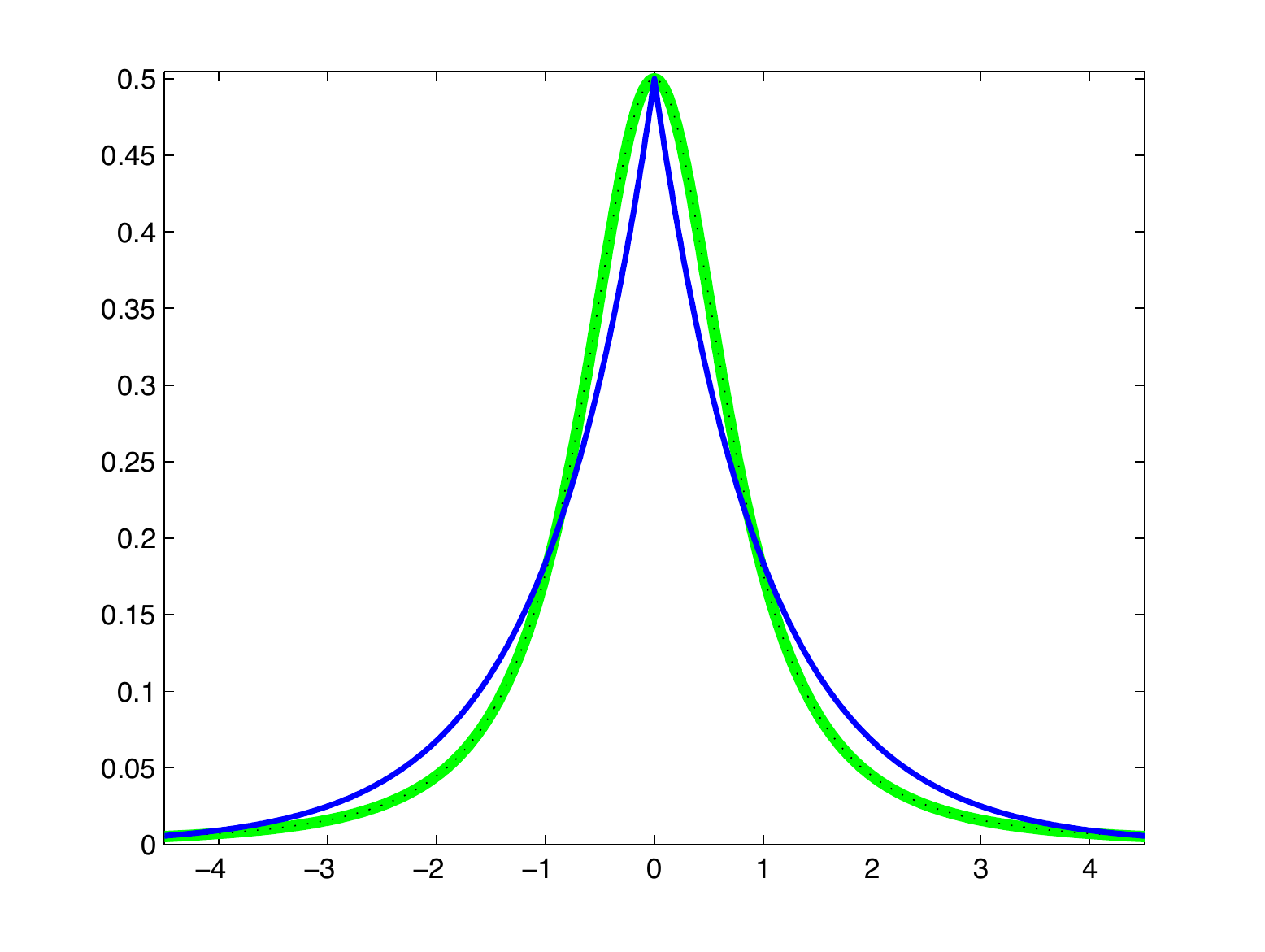} \hspace{0.5cm}
\includegraphics[width=0.4\textwidth]{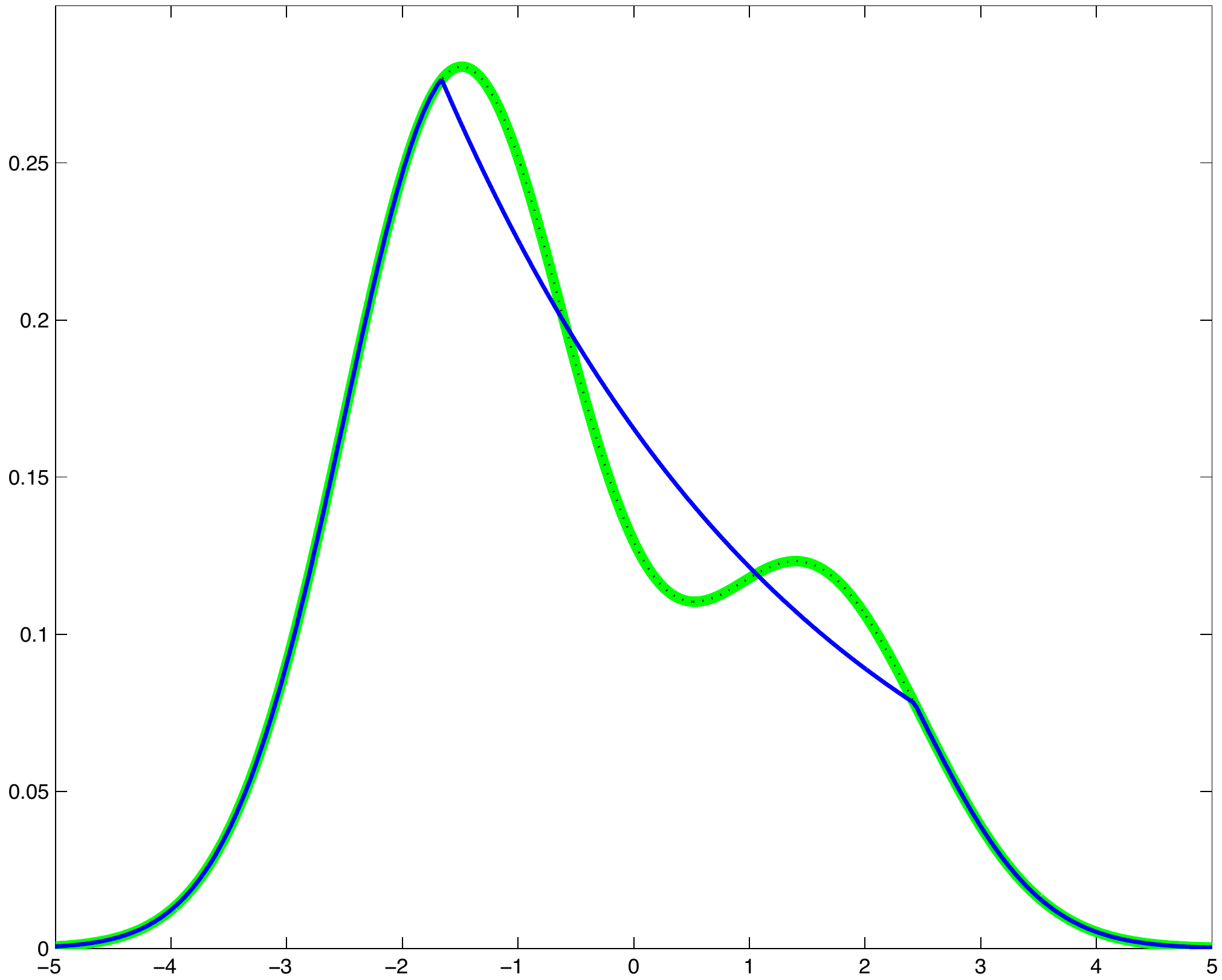}
\caption{\label{Fig:Projs}Left: the scaled $t_2$ density $f(x) = (1+x^2)^{-3/2}/2$ (green) and its Laplace log-concave projection $f^*(x) = e^{-|x|}/2$ (blue).  Right: the density of the normal mixture $0.7N(-1.5,1) + 0.3N(1.5,1)$ (green) together with its log-concave projection (blue); the normal mixture satisfies the conditions of Proposition~\ref{Prop:ConcConv}.}
\end{figure}

\section{Stronger forms of convergence and consistency}

In minor abuse of standard notation, if $(f_n), f$ are densities on $\mathbb{R}^d$, we write $f_n \stackrel{d}{\rightarrow} f$ to mean $\int_{\mathbb{R}^d} g(x)f_n(x) \, dx \rightarrow \int_{\mathbb{R}^d} g(x)f(x) \, dx$ for all bounded continuous functions $g:\mathbb{R}^d \rightarrow \mathbb{R}$.  The constraint of log-concavity rules out certain pathologies and means we can strengthen certain convergence statements:
\begin{thm}[\citet{CuleSamworth2010,SHD2011}]
Let $(f_n)$ be a sequence in $\mathcal{F}_d$ with $f_n \stackrel{d}{\rightarrow} f$ for some density $f$ on $\mathbb{R}^d$.  Then $f$ is log-concave.  Moreover, if $\alpha_0 > 0$ and $\beta_0 \in \mathbb{R}$ are such that $f(x) \leq e^{-\alpha_0\|x\|+\beta_0}$ for all $x \in \mathbb{R}^d$, then for all $\alpha < \alpha_0$,
\[
\int_{\mathbb{R}^d} e^{\alpha\|x\|}|f_n(x) - f(x)| \, dx \rightarrow 0
\]
as $n \rightarrow \infty$.
\end{thm}
Thus, in the presence of log-concavity, convergence in distribution statements automatically yield convergence in certain exponentially weighted total variation distances.  

A very natural question about log-concave projections, with important implications for the consistency of the log-concave maximum likelihood estimator, is `In what sense does a distribution $Q \in \mathcal{P}_d$ need to be close to $P \in \mathcal{P}_d$ in order for $\psi^*(Q)$ to be close to $\psi^*(P)$'?  To answer this, we first recall that the Mallows-1 distance\footnote{Also known as the Wasserstein distance, Monge--Kantorovich distance and Earth Mover's distance.} $d_1$ between probability measures $P,Q$ on $\mathbb{R}^d$ with finite first moment is given by
\[
d_1(P,Q) := \inf_{(X,Y) \sim (P,Q)} \mathbb{E}\|X-Y\|,
\]
where the infimum is taken over all pairs of random vectors $(X,Y)$ defined on the same probability space with $X \sim P$ and $Y \sim Q$.  It is well-known that $d_1(P_n,P) \rightarrow 0$ if and only if both $P_n \stackrel{d}{\rightarrow} P$ and $\int_{\mathbb{R}^d} \|x\| \, dP_n(x) \rightarrow \int_{\mathbb{R}^d} \|x\| \, dP(x)$.  
\begin{thm}[\citet{DSS2011}]
\label{Thm:Cont}
Suppose that $P \in \mathcal{P}_d$ and that $d_1(P_n,P) \rightarrow 0$.  Then $L^*(P_n) \rightarrow L^*(P)$, $P_n \in \mathcal{P}_d$ for sufficiently large $n$, and, taking $\alpha_0 > 0$ and $\beta_0 \in \mathbb{R}$ such that $\psi^*(P)(x) \leq e^{-\alpha_0\|x\|+\beta_0}$ for all $x \in \mathbb{R}^d$, we have for $\alpha < \alpha_0$ that

\[
\int_{\mathbb{R}^d} e^{\alpha\|x\|}|\psi^*(P_n)(x) - \psi^*(P)(x)| \, dx \rightarrow 0
\]
as $n \rightarrow \infty$.
\end{thm}
The Mallows convergence cannot in general be weakened to $P_n \stackrel{d}{\rightarrow} P$.  In particular, if $P = U\{-1,1\}$ and $P_n = (1-n^{-1})U\{-1,1\} + n^{-1}U\{-(n+1),n+1\}$, then $P_n \stackrel{d}{\rightarrow} P$ but it can be shown that
\[
\int_{-\infty}^\infty |\psi^*(P_n) - \psi^*(P)| \rightarrow \frac{4}{5^{1/2}+1}.
\]
Writing $d_{\mathrm{TV}}(f,g) := \frac{1}{2}\int_{\mathbb{R}^d} |f-g|$, Theorem~\ref{Thm:Cont} implies that the log-concave projection $\psi^*$ is continuous when considered as a map between the metric spaces $(\mathcal{P}_d,d_1)$ and $(\mathcal{F}_d,d_{\mathrm{TV}})$.  However, it is not uniformly continuous: for instance, let $P_n = U[-1/n,1/n]$ and $Q_n = U[-1/n^2,1/n^2]$.  Then $d_1(P_n,Q_n) = \frac{1}{2n} - \frac{1}{2n^2} \rightarrow 0$, but
\[
d_{\mathrm{TV}}\bigl(\psi^*(P_n),\psi^*(Q_n)\bigr) = \int_0^{1/n^2} \frac{n^2}{2} - \frac{n}{2} \, dx + \int_{1/n^2}^{1/n} \frac{n}{2} \, dx \rightarrow 1.
\]
One of the great advantages of working in the general framework of log-concave projections for arbitrary $P \in \mathcal{P}_d$, as opposed to simply focusing on empirical distributions, is that one can study analytical properties of the projection as above, meaning that the only probabilistic arguments required to deduce convergence statements about the log-concave maximum likelihood estimator are simple facts about the convergence of the empirical distribution.  This is illustrated in the following corollary.
\begin{corollary}[\citet{DSS2011}]
\label{Cor:Cons}
Suppose that $X_1,X_2,\ldots$ are independent and identically distributed with distribution $P \in \mathcal{P}_d$, and let $\mathbb{P}_n$ denote the empirical distribution of $X_1,\ldots,X_n$.  Then, with probability one, $\hat{f}_n := \psi^*(\mathbb{P}_n)$ is well-defined for sufficiently large $n$, and taking $\alpha_0 > 0$ and $\beta_0 \in \mathbb{R}$ such that $f^*(x) := \psi^*(P)(x) \leq e^{-\alpha_0\|x\|+\beta_0}$ for all $x \in \mathbb{R}^d$, we have for $\alpha < \alpha_0$ that
\[
\int_{\mathbb{R}^d} e^{\alpha\|x\|}|\hat{f}_n(x) - f^*(x)| \, dx \stackrel{\mathrm{a.s.}}{\rightarrow} 0
\]
as $n \rightarrow \infty$.
\end{corollary}
\begin{proof}
let $\mathcal{H} := \bigl\{h: \mathbb{R}^d \rightarrow [-1,1]:|h(x) - h(y)| \leq \|x-y\| \ \text{for all} \ x,y \in \mathbb{R}^d\bigr\}$, and define the bounded Lipschitz distance between probability measures $P$ and $Q$ on $\mathbb{R}^d$ by
\[
d_{\mathrm{BL}}(P,Q) := \sup_{h \in \mathcal{H}} \int_{\mathbb{R}^d} h \, d(P-Q).
\]
Then $d_{\mathrm{BL}}$ metrises convergence in distribution for probability measures on $\mathbb{R}^d$, and from Varadarajan's theorem \citep[][Theorem~11.4.1]{Dudley2002}, we deduce that $d_{\mathrm{BL}}(\mathbb{P}_n,P) \stackrel{\mathrm{a.s.}}{\rightarrow} 0$.  In particular, since the set of probability measures $P$ on $\mathbb{R}^d$ with $P(H) < 1$ for all hyperplanes $H$ is an open subset of the set of all probability measures on $\mathbb{R}^d$ in the topology of weak convergence \citep[][Lemma~2.13]{DSS2011}, it follows that with probability one, $\mathbb{P}_n \in \mathcal{P}_d$ for sufficiently large $n$, and $\hat{f}_n$ is well-defined for such $n$.

Since we also have $\int_{\mathbb{R}^d} \|x\| \, d\mathbb{P}_n(x) \stackrel{\mathrm{a.s.}}{\rightarrow} \int_{\mathbb{R}^d} \|x\| \, dP(x)$ by the strong law of large numbers, it follows that $d_1(\mathbb{P}_n,P) \stackrel{\mathrm{a.s.}}{\rightarrow} 0$.  The second part of the result therefore follows by Theorem~\ref{Thm:Cont}.
\end{proof}
Corollary~\ref{Cor:Cons} yields the (strong) consistency of the log-concave maximum likelihood estimator in exponentially weighted total variation distances, and also provides a robustness to misspecification guarantee in the case where the true distribution $P$ does not have a log-concave density.

\section{Rates of convergence and adaptation}

Historically, a great deal of effort has gone into understanding rates of convergence in shape-constrained estimation problems, with both local (pointwise) and global rates being considered.  For the log-concave maximum likelihood estimator, the following result, a special case of \citet[][Theorem~2.1]{BRW2009}, establishes the pointwise rates of convergence in the case $d=1$:
\begin{thm}[\citet{BRW2009}]
Let $X_1,\ldots,X_n \stackrel{\mathrm{iid}}{\sim} f_0 \in \mathcal{F}_1$, let $f_0(x_0) > 0$ and suppose that $\phi_0 := \log f_0$ is twice continuously differentiable in a neighbourhood of $x_0$ with $\phi_0''(x_0) < 0$.  Let $W$ be a standard two-sided Brownian motion on $\mathbb{R}$, and let 
\[
Y(t) := \left\{ \begin{array}{ll} \int_0^t W(s) \, ds - t^4 & \mbox{for $t \geq 0$} \\
\int_t^0 W(s) \, ds - t^4 & \mbox{for $t < 0$.} \end{array} \right.
\]
Then the log-concave maximum likelihood estimator $\hat{f}_n$ satisfies
\begin{equation}
\label{Eq:Limit}
n^{2/5}\{\hat{f}_n(x_0) - f_0(x_0)\} \stackrel{d}{\rightarrow} \biggl(\frac{f_0(x_0)^3|\phi_0''(x_0)|}{24}\biggr)^{1/5}H''(0),
\end{equation}
where $\{H(t):t \in \mathbb{R}\}$ is the `lower invelope' process of $Y$, so that $H(t) \leq Y(t)$ for all $t \in \mathbb{R}$, $H''$ is concave and $H(t) = Y(t)$ if the slope of $H''$ decreases strictly at~$t$.
\end{thm}
The non-standard limiting distribution is characteristic of shape-constrained estimation problems.  \citet{BRW2009} study the more general case where more than two derivatives of $\phi_0$ may vanish at $x_0$, in which case a faster rate is obtained; they also study the joint convergence of $\hat{f}_n$ with its derivative $\hat{f}_n'$.  The pointwise convergence rate in $d$ dimensions remains an open problem, though \citet{SereginWellner2010} obtained a minimax lower bound for pointwise estimation at $x_0$ with respect to absolute error loss of order $n^{-2/(d+4)}$, provided $\phi_0$ is twice continuously differentiable in a neighbourhood of $x_0$ and the determinant of the Hessian matrix of $\phi_0$ at $x_0$ does not vanish.  This is the familiar rate attained by, e.g.~kernel density estimators, under similar smoothness conditions but without the log-concavity assumption.

An interesting feature of~\eqref{Eq:Limit} is that the limiting distribution depends in a complicated way on the unknown true density.  This makes it challenging to apply this result directly to construct confidence intervals for $f_0(x_0)$.  However, in the special case where $x_0$ is the mode of $f_0$, \citet{DossWellner2016} have recently proposed an approach for confidence interval construction based on comparing the log-concave MLE at $x_0$ with the constrained MLE where the mode of the density is fixed at $m$, say.  The key observation is that, under the null hypothesis, the likelihood ratio statistic is asymptotically pivotal.

We now turn to global rates of convergence, and write $d_{\mathrm{H}}^2(f,g) := \int_{\mathbb{R}^d} (f^{1/2} - g^{1/2})^2$ for the squared Hellinger distance between densities $f$ and $g$.  The same rate as for pointwise estimation had been expected in the light of the facts that any concave function on $\mathbb{R}^d$ is twice differentiable (Lebesgue) almost everywhere in its domain \citep{Aleksandrov1939}, and that for twice continuously differentiable functions, concavity is equivalent to a second derivative condition, namely that the Hessian matrix is non-positive definite.  The following minimax lower bound therefore came as a surprise:
\begin{thm}[\citet{KimSamworth2016}]
\label{Thm:LB}
Let $X_1,\ldots,X_n \stackrel{\mathrm{iid}}{\sim} f_0 \in \mathcal{F}_d$, and let $\tilde{F}_n$ denote the set of all estimators of $f_0$ based on $X_1,\ldots,X_n$.  Then for each $d \in \mathbb{N}$, there exists $c_d > 0$ such that
\[
\inf_{\tilde{f}_n \in \tilde{\mathcal{F}}_n} \sup_{f_0 \in \mathcal{F}_d} \mathbb{E}_{f_0}d_{\mathrm{H}}^2(\tilde{f}_n,f_0) \geq \left\{ \begin{array}{ll} c_1 n^{-4/5} & \mbox{if $d=1$} \\
c_dn^{-2/(d+1)} & \mbox{if $d \geq 2$.} \end{array} \right.
\]
\end{thm}
Theorem~\ref{Thm:LB} yields the expected lower bound when $d=1,2$ (note that $2/(d+1) = 4/(d+4) = 2/3$ when $d=2$).  However, it also reveals that log-concave density estimation in three or more dimensions is fundamentally more challenging in this minimax sense than estimating a density with two bounded derivatives.  The reason is that although log-concave densities are twice differentiable almost everywhere, they can be badly behaved (in particular, discontinuous) on the boundary of their support; recall that uniform densities on convex, compact sets in $\mathbb{R}^d$ belong to $\mathcal{F}_d$.  It turns out that it is the difficulty of estimating the support of the density that drives the rate in these higher dimensions. 

The following complementary result provides the corresponding global rate of convergence for the log-concave MLE in squared Hellinger distance in low-dimensional cases.
\begin{thm}[\citet{KimSamworth2016}]
\label{Thm:UB}
Let $X_1,\ldots,X_n \stackrel{\mathrm{iid}}{\sim} f_0 \in \mathcal{F}_d$, and let $\hat{f}_n$ denote the log-concave MLE based on $X_1,\ldots,X_n$.  Then
\[
\sup_{f_0 \in \mathcal{F}_d} \mathbb{E}_{f_0}d_{\mathrm{H}}^2(\tilde{f}_n,f_0) = \left\{ \begin{array}{ll} O(n^{-4/5}) & \mbox{if $d=1$} \\
O(n^{-2/3}\log n) & \mbox{if $d=2$} \\
O(n^{-1/2}\log n) & \mbox{if $d=3$}. \end{array} \right.
\]
\end{thm}
Thus the log-concave MLE attains the minimax optimal rate in terms of squared Hellinger risk when $d=1$, and attains the minimax optimal rate up to logarithmic factors when $d=2,3$.  The proofs of these results rely on empirical process theory and delicate bracketing entropy bounds for the relevant class of log-concave densities, made more complicated by the fact that the domains of the log-densities can be an arbitrary $d$-dimensional closed, convex set.  The argument proceeds by approximating these domains by convex polygons, which can be triangulated into simplices, and appropriate bracketing entropy bounds for concave functions on such domains are known \citep[e.g.][]{GaoWellner2015}.  Critically, when $d \leq 3$, a convex polygon with $m$ vertices can be triangulated into $O(m)$ simplices; however, when $d \geq 4$, such results from discrete convex geometry are not available, which explains why no rate of convergence has yet been obtained in such cases.  We mention, however, that lower bounds on the bracketing entropy obtained in \citet{KimSamworth2016} strongly suggest, but do not prove, that the log-concave MLE will be rate-suboptimal when $d \geq 4$.

Although Theorem~\ref{Thm:UB} provides strong guarantees on the worst case performance of the log-concave MLE in low-dimensional cases, it ignores one of the appealing features of the estimator, namely its potential to adapt to certain characteristics of the unknown true density.  \citet{DumbgenRufibach2009} obtained the first such result in the case $d=1$.  Recall that given an interval $I$, $\beta \in [1,2]$ and $L > 0$, we say $h:\mathbb{R} \rightarrow \mathbb{R}$ belongs to the H\"older class $\mathcal{H}_{\beta,L}(I)$ if for all $x,y \in I$, we have
\begin{alignat*}{2}
 |h(x) - h(y)| &\leq L|x-y|,  &&\text{if} \ \beta=1 \\
|h'(x) - h'(y)| &\leq L|x-y|^{\beta-1}, \quad &&\text{if} \ \beta > 1.
\end{alignat*}
\begin{thm}[\citet{DumbgenRufibach2009}]
Let $X_1,\ldots,X_n \stackrel{\mathrm{iid}}{\sim} f_0 \in \mathcal{F}_1$, and assume that $\phi_0 := \log f_0 \in \mathcal{H}_{\beta,L}(I)$ for some $\beta \in [1,2]$, $L > 0$ and compact interval $I \subseteq \mathrm{int}\bigl(\mathrm{dom}(\phi_0)\bigr)$.  Then
\[
\sup_{x_0 \in I} |\hat{f}_n(x_0) - f_0(x_0)| = O_p\biggl(\Bigl(\frac{\log n}{n}\Bigr)^{\beta/(2\beta+1)}\biggr).
\]
\end{thm}
Here the log-concave MLE is adapting to unknown smoothness.  When measuring loss in the supremum norm, the need to restrict attention to a compact interval in the interior of support of $f_0$ is suggested by the right-hand plot in Figure~\ref{Fig:1D}.    

Other adaptation results are motivated by the thought that since the log-concave MLE is piecewise affine, we might hope for faster rates of convergence in cases where $\log f_0$ is made up of a relatively small number of affine pieces.  We now describe two such results.  For $k \in \mathbb{N}$ we define $\mathcal{F}^k$ to be the class of log-concave densities $f$ on $\mathbb{R}$ for which $\log f$ is $k$-affine in the sense that there exist intervals $I_1,\ldots,I_k$ such that $f$ is supported on $I_1 \cup \ldots \cup I_k$, and $\log f$ is affine on each~$I_j$.  In particular, densities in $\mathcal{F}^1$ are uniform or (possibly truncated) exponential, and can be parametrised as
\[
f_{\alpha,s_1,s_2}(x) := \left\{ \begin{array}{ll} \frac{1}{s_2-s_1}\mathbbm{1}_{\{x \in [s_1,s_2]\}} & \mbox{if $\alpha=0$} \\
\frac{\alpha}{e^{\alpha s_2}-e^{\alpha s_1}} e^{\alpha x} \mathbbm{1}_{\{x \in [s_1,s_2]\}} & \mbox{if $\alpha \neq 0$,} \end{array} \right. 
\]
for $(\alpha,s_1,s_2) \in \mathcal{T} := (\mathbb{R} \times \mathcal{T}_0) \, \cup \, \bigl((0,\infty) \times \{-\infty\} \times \mathbb{R}\bigr) \, \cup \, \bigl((-\infty,0) \times \mathbb{R} \times \{\infty\}\bigr)$, where $\mathcal{T}_0 := \{(s_1,s_2) \in \mathbb{R}^2:s_1 < s_2\}$.  Define a continuous, strictly increasing function $\rho:\mathbb{R} \rightarrow (0,\infty)$ by
\begin{equation}
\label{Eq:rho}
\rho(x) := \left\{ \begin{array}{ll} \frac{2e^x(x-1)-x^2+2}{2e^x-2-2x-x^2} & \mbox{if $x \neq 0$} \\
2 & \mbox{if $x = 0$}; \end{array} \right.
\end{equation}
cf.~Figure~\ref{Fig:rho}.  It can be shown that $\rho(x) \leq \max\{\rho(2),\rho(x)\} \leq \max(3,2x)$ for all $x \in \mathbb{R}$.
\begin{figure}
\includegraphics[width=0.5\textwidth]{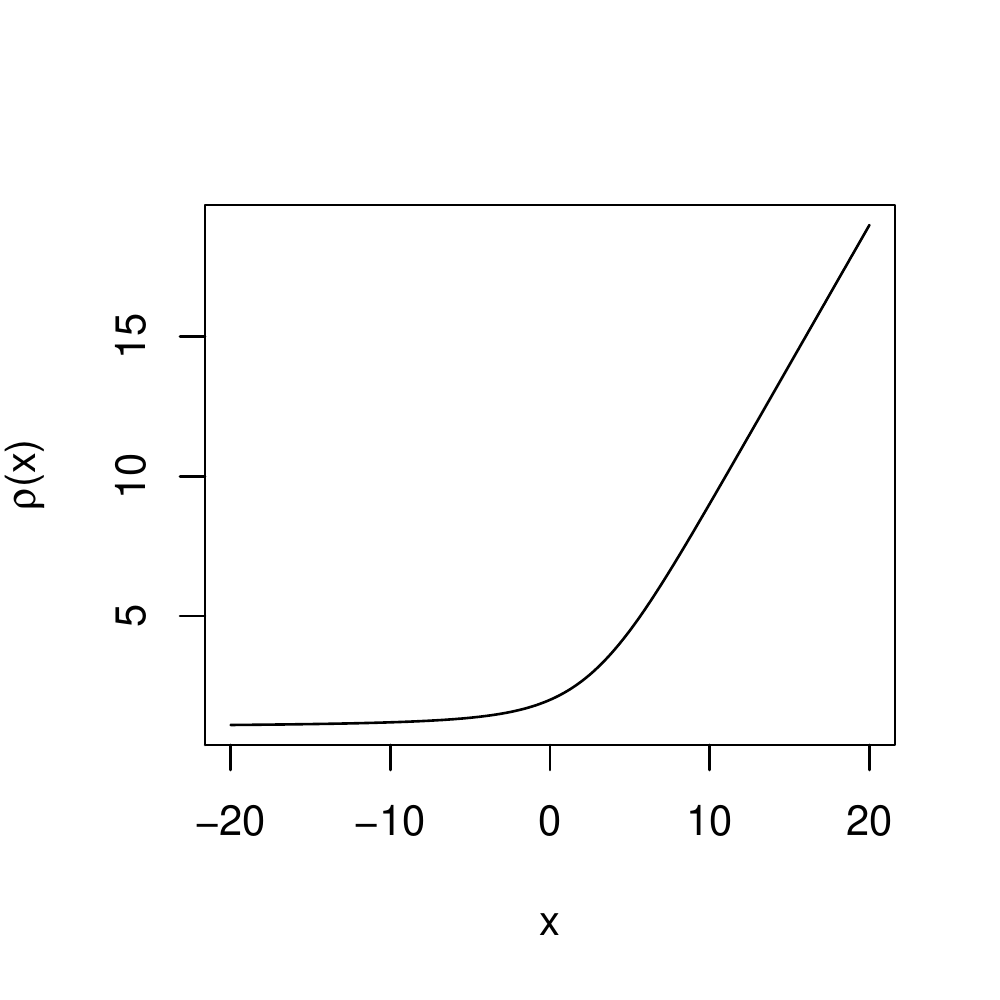}
\caption{\label{Fig:rho}The function $\rho$ defined in~\eqref{Eq:rho}.} 
\end{figure}
\begin{thm}[\citet{KGS2016}]
\label{Thm:Adap1}
Let $X_1,\ldots,X_n \stackrel{\mathrm{iid}}{\sim} f_{\alpha,s_1,s_2} \in \mathcal{F}^1$ with $n \geq 5$, and let $\hat{f}_n$ denote the log-concave MLE.  Then, writing $\kappa^* := \alpha(s_2-s_1)$,
\[
\mathbb{E}_{f_0}d_{\mathrm{TV}}(\hat{f}_n,f_0) \leq \frac{\min\{2\rho(|\kappa^*|),6\log n\}}{n^{1/2}}.
\]
\end{thm}
In fact, Theorem~\ref{Thm:Adap1} is a special case of the result given in~\citet{KGS2016}, which allows the true density $f_0$ to be arbitrary, and includes an additional approximation error term that measures the proximity of $f_0$ to the class $\mathcal{F}^1$.  An important consequence of~Theorem~\ref{Thm:Adap1} is the fact that if $|\alpha|$ is small, then the log-concave MLE can attain the parametric rate of convergence in total variation distance.  In particular, if $f_0$ is a uniform density on a compact interval (so that $\kappa^* = 0$), then $\mathbb{E}_{f_0}d_{\mathrm{TV}}(\hat{f}_n,f_0) \leq 4/n^{1/2}$; cf.~the right-hand plot of Figure~\ref{Fig:1D} again.  Interestingly, this behaviour is in stark constrast to that of the least squares convex regression estimator with respect to squared error loss in the random design problem where covariates are uniformly distributed on $[0,1]$ and the responses are uniform on $\{-1,1\}$: in that case, the regression function is zero, but the risk of the estimator is infinite \citep{BGS2015}!  The proof of Theorem~\ref{Thm:Adap1} relies on a version of Marshall's inequality for log-concave density estimation.  A special case of this result states that if $X_1,\ldots,X_n \stackrel{\mathrm{iid}}{\sim} f_{\alpha,s_1,s_2} \in \mathcal{F}^1$, then writing $X_{(1)} := \min_i X_i$, $X_{(n)} := \max_i X_i$ and $\kappa := \alpha(X_{(n)} - X_{(1)})$, we have
\begin{equation}
\label{Eq:Marshall}
\sup_{x \in \mathbb{R}} |\hat{F}_n(x) - F_0(x)| \leq \rho(|\kappa|)\sup_{x \in \mathbb{R}}|\mathbb{F}_n(x) - F_0(x)|,
\end{equation}
where $F_0$ and $\hat{F}_n$ denote the distribution functions corresponding to the true density and the log-concave MLE respectively, and where $\mathbb{F}_n$ denotes the empirical distribution function\footnote{The original Marshall's inequality \citep{Marshall1970} applies to the (integrated) Grenander estimator, in which context $\rho(|\kappa|)$ in~\eqref{Eq:Marshall} may be replaced by 1.}.   

We now aim to generalise these ideas to situations where $f_0$ is close to $\mathcal{F}^k$, but assume only that $X_1,\ldots,X_n \stackrel{\mathrm{iid}}{\sim} f_0 \in \mathcal{F}_1$.  An application of Lemma~\ref{Lemma:Basic} to the function $\Delta(x) = \log \frac{f_0(x)}{\hat{f}_n(x)}$ yields
\[
d_{\mathrm{KL}}^2(\hat{f}_n,f_0) \leq \frac{1}{n}\sum_{i=1}^n \log \frac{\hat{f}_n(X_i)}{f_0(X_i)} =: d_X^2(\hat{f}_n,f_0).
\]
In particular, an upper bound on $d_X^2(\hat{f}_n,f_0)$ immediately provides corresponding bounds on $d_{\mathrm{TV}}^2(\hat{f}_n,f_0)$, $d_{\mathrm{H}}^2(\hat{f}_n,f_0)$ and $d_{\mathrm{KL}}^2(\hat{f}_n,f_0)$.
\begin{thm}[\citet{KGS2016}]
There exists a universal constant $C > 0$ such that for $n \geq 2$,
\[
\mathbb{E}_{f_0}d_X^2(\hat{f}_n,f_0) \leq \min_{k=1,\ldots,n} \biggl\{\frac{Ck}{n}\log^{5/4} \frac{en}{k} + \inf_{f_k \in \mathcal{F}^k} d_{\mathrm{KL}}^2(f_0,f_k)\biggr\}.
\]
\end{thm}
To help understand this theorem, first consider the case where $f_0 \in \mathcal{F}^k$.  Then $\mathbb{E}_{f_0}d_X^2(\hat{f}_n,f_0) \leq \frac{Ck}{n}\log^{5/4} (en/k)$, which is nearly the parametric rate when $k$ is small.  More generally, this rate holds when $f_0 \in \mathcal{F}_1$ is only close to $\mathcal{F}^k$ in the sense that the approximation error $d_{\mathrm{KL}}^2(f_0,f_k)$ is $O\bigl(\frac{k}{n}\log^{5/4} \frac{en}{k}\bigr)$.  The result is known as a `sharp' oracle inequality, because the leading constant for this approximation error term is 1.  See also \citet{BaraudBirge2016}, who also obtain an oracle inequality for their general $\rho$-estimation procedure.  It is worth noting that the techniques of proof, which rely on empirical process theory and local bracketing entropy bounds, are completely different from those used in the proof of Theorem~\ref{Thm:Adap1}.

\section{Higher-dimensional problems}
\label{Sec:ICA}

The minimax lower bound in Theorem~\ref{Thm:LB} is relatively discouraging for the prospects of log-concave density estimation in higher dimensions.  It is natural, then, to consider additional structures that reduce the complexity of the class $\mathcal{F}_d$, thereby increasing the potential for applications outside low-dimensional settings.  The purpose of this section is two explore two ways of imposing such structures, namely through independence and symmetry constraints.

In the simplest, noiseless case of Independent Component Analysis (ICA), one observes independent replicated of a random vector 
\begin{equation}
\label{Eq:ICA}
X := AS,
\end{equation}
where $A \in \mathbb{R}^{d \times d}$ is a deterministic, invertible matrix, and $S$ is a $d$-dimensional random vector with independent components.  One can think of the model as being the density estimation analogue of mulitple index models in regression.  ICA models have found an enormous range of applications across signal processing, machine learning and medical imaging, to name just a few; see \citet{HKO2001} for an introduction to the field.  The main interest is in estimating the unmixing matrix $W := A^{-1}$, with estimation of the marginal distributions of the components of $S$ as a secondary goal.  Let $\mathcal{W}$ denote the set of all invertible $d \times d$ real matrices, let $\mathcal{B}_d$ denote the set of all Borel subsets of $\mathbb{R}^d$, and let $\mathcal{P}_d^{\mathrm{ICA}}$ denote the set of $P \in \mathcal{P}_d$ with
\[
P(B) = \prod_{j=1}^d P_j(w_j^\top B) \quad \forall B \in \mathcal{B}_d, 
\]
for some $W = (w_1,\ldots,w_d)^\top \in \mathcal{W}$ and $P_1,\ldots,P_d \in \mathcal{P}_1$.  Thus $\mathcal{P}_d^{\mathrm{ICA}}$ is the set of distributions of random vectors $X$ with $\mathbb{E}(\|X\|) < \infty$ satisfying~\eqref{Eq:ICA}.  As stated, the model~\eqref{Eq:ICA} is not identifiable, as we can write $X = ADPP^\top D^{-1}S$, where $D$ is a diagonal $d \times d$ matrix with non-zero diagonal entries, and $P \in \mathbb{R}^{d \times d}$ is a permutation matrix (note that $ADP$ is invertible and $P^\top D^{-1}S$ has independent components).  Fortunately, these can be regarded as `trivial' lack of identifiability problems, because it is typically the directions of the set of rows of $W := A^{-1}$ that are of interest, not their order or magnitude.  \citet{ErikssonKoivunen2004} proved that the pair of conditions that none of $P_1,\ldots,P_d$ are Dirac point masses and at most one of them is Gaussian is necessary and sufficient for the ICA model to be identifiable up to the permutation and scaling transformations described above. 

Now let $\mathcal{F}_d^{\mathrm{ICA}}$ denote the set of $f \in \mathcal{F}_d$ with
\[
f(x) = |\det W|\prod_{j=1}^d f_j(w_j^\top x)
\]
for some $W = (w_1,\ldots,w_d)^\top \in \mathcal{W}$ and $f_1,\ldots,f_d \in \mathcal{F}_1$.  In this way, $\mathcal{F}_d^{\mathrm{ICA}}$ is the set of densities of random vectors $X$ satisfying~\eqref{Eq:ICA}, where each component of $S$ has a log-concave density.  Define the log-concave ICA projection on $\mathcal{P}_d$ by
\[
\psi^{**}(P) := \argmax_{f \in \mathcal{F}_d^{\mathrm{ICA}}} \int_{\mathbb{R}^d} \log f \, dP.
\]
In general, $\psi^{**}(P)$ only defines a non-empty, proper subset of $\mathcal{F}_d^{\mathrm{ICA}}$ rather than a unique element.  However, the following theorem gives uniqueness in an important special case, and the form of the log-concave ICA projection here is key to the success of this approach to fitting ICA models.
\begin{thm}[\citet{SamworthYuan2012}]
\label{Thm:ICAProj}
If $P \in \mathcal{P}_d^{\mathrm{ICA}}$, then $\psi^{**}(P)$ defines a unique element of $\mathcal{F}_d^{\mathrm{ICA}}$.  In fact, the restrictions of $\psi^{**}$ and $\psi^*$ to $\mathcal{P}_d^{\mathrm{ICA}}$ coincide.  Moreover, suppose that $P \in \mathcal{P}_d^{\mathrm{ICA}}$, so
\[
P(B) = \prod_{j=1}^d P_j(w_j^\top B) \quad \forall B \in \mathcal{B}_d, 
\]
for some $W = (w_1,\ldots,w_d)^\top \in \mathcal{W}$ and $P_1,\ldots,P_d \in \mathcal{P}_1$.  Then $f^{**} := \psi^{**}(P)$ can be written explicitly as
\[
f^{**}(x) = |\det W|\prod_{j=1}^d f_j^*(w_j^\top x),
\]
where $f_j^* := \psi^*(P_j)$.
\end{thm}
The fact that $\psi^*$ preserves the ICA structure is a consequence of Lemma~\ref{Lemma:Affine} and Proposition~\ref{Prop:Indep}.  However, the most interesting aspect of this result is the fact that the unmixing matrix $W$ is preserved by the log-concave projection.  This suggests that, at least from the point of view of estimating $W$, there is no loss of generality in assuming that the marginal distributions of the components of $S$ have log-concave densities provided they have finite means.  Another crucial result is the fact that the log-concave ICA projection of $P \in \mathcal{P}_d^{\mathrm{ICA}}$ does not sacrifice identifiability: in fact, $\psi^{**}(P)$ is identifiable if and only if $P$ is identifiable. 

Given data $X_1,\ldots,X_n \stackrel{\mathrm{iid}}{\sim} P \in \mathcal{P}_d$ with empirical distribution $\mathbb{P}_n$, we can therefore fit an ICA model by computing $\hat{f}_n := \psi^{**}(\mathbb{P}_n)$.  This estimator has similar consistency properties to the original log-concave projection, and requires the maximisation of   
\[
\ell(W,f_1,\ldots,f_d;X_1,\ldots,X_n) := \log |\det W| + \frac{1}{n} \sum_{i=1}^n \sum_{j=1}^d \log f_j(w_j^\top X_i)
\]
over $W \in \mathcal{W}$ and $f_1,\ldots,f_d \in \mathcal{F}_1$.  For reasons of numerical stability, however, it is convenient to `pre-whiten' the estimator by setting $Z_i := \hat{\Sigma}^{-1/2}X_i$ for $i=1,\ldots,n$, where $\hat{\Sigma}$ denotes the sample covariance matrix.  We can then instead obtain a maximiser $(\hat{O},\hat{g}_1,\ldots,\hat{g}_d)$ of $\ell(O,g_1,\ldots,g_d;Z_1,\ldots,Z_n)$ over $O \in O(d)$, the set of $d \times d$ orthogonal matrices, and $g_1,\ldots,g_d \in \mathcal{F}_1$, before setting $\hat{\hat{W}} := \hat{O}\hat{\Sigma}^{-1/2}$ and $\hat{\hat{f}}_j := \hat{g}_j$.  This estimator has the same consistency properties as the original proposal, provided that $\int_{\mathbb{R}^d} \|x\|^2 \, dP(x) < \infty$.  In effect, it breaks down the estimation of the $d^2$ parameters in $W$ into two stages: first, we use $\hat{\Sigma}$ to estimate the $d(d + 1)/2$ free parameters of the symmetric, positive definite matrix $\Sigma$, leaving only the maximisation over the $d(d - 1)/2$ free parameters of $O \in O(d)$ at the second stage.  Even after pre-whitening, however, there is an additional computational challenge relative to the orginal log-concave MLE caused by the fact that the objective function $\ell$ is only bi-concave\footnote{In other words, $\ell$ is concave in $O$ for fixed $g_1,\ldots,g_d$, and concave in $g_1,\ldots,g_d$ for fixed $O$.} in $O$ and $g_1,\ldots,g_d$, but not jointly concave in these arguments.  Since we only have to deal with computation of univariate log-concave maximum likelihood estimators, however, marginal updates are straightforward, and taking the solution with highest log-likelihood over several random initial values for the variables can lead to satisfactory solutions \citep{SamworthYuan2012}.

Symmetry constraints provide another alternative approach to extending the scope of shape-constrained methods to higher dimensions.  For simplicity of exposition, we focus on the simplest case of spherical symmetry, as studied recently by \citet{XuSamworth2017}, though more general symmetry constraints may also be considered.  We write $\mathcal{F}_d^{\mathrm{SS}}$ for the set of spherically symmetric $f \in \mathcal{F}_d$, and let $\Phi^{\mathrm{SS}}$ denote the class of upper semi-continuous, decreasing, concave functions $\phi:[0,\infty) \rightarrow [-\infty,\infty)$.  The starting point for the symmetry-based approach is the observation that a density $f$ on $\mathbb{R}^d$ belongs to $\mathcal{F}_d^{\mathrm{SS}}$ if and only if $f(x) = e^{\phi(\|x\|)}$ for some $\phi \in \Phi^{\mathrm{SS}}$.  One can then define the notion of spherically symmetric log-concave projection, which has several similarities with the theory presented in Sections~\ref{Sec:LCProjections} and~\ref{Sec:Props} (though with some notable differences, especially with regard to moment preservation properties).  In particular, given data $X_1,\ldots,X_n \in \mathbb{R}^d$ that are not all zero, there exists a unique spherically spherically log-concave MLE $\hat{f}_n^{\mathrm{SS}}$.  This estimator can be computed using a variant of the Active Set algorithm outlined in Section~\ref{Sec:Algorithm}.  Importantly, this algorithm only depends on $d$ through the need to compute $Z_i := \|X_i\|$ for $i=1,\ldots,n$ at the outset, and it therefore scales extremely well to high-dimensional cases, even when $d$ may be in the hundreds of thousands.

The following worst case bound reveals that $\hat{f}_n^{\mathrm{SS}}$ succeeds in evading the curse of dimensionality:
\begin{thm}[\citet{XuSamworth2017}]
Let $f_0 \in \mathcal{F}_d^{\mathrm{SS}}$, let $X_1,\ldots,X_n \stackrel{\mathrm{iid}}{\sim} f_0$, and let $\hat{f}_n^{\mathrm{SS}}$ denote the corresponding spherically symmetric log-concave MLE.  Then there exists a universal constant $C > 0$ such that
\[
\sup_{f_0 \in \mathcal{F}_d^{\mathrm{SS}}} \mathbb{E}d_X^2(\hat{f}_n^{\mathrm{SS}},f_0) \leq Cn^{-4/5}.
\]
\end{thm}
Similar to the ordinary log-concave MLE, we have $d_X^2(\hat{f}_n^{\mathrm{SS}},f_0) \geq d_{\mathrm{KL}}^2(\hat{f}_n^{\mathrm{SS}},f_0)$, and the interesting feature of this bound is that it does not depend on $d$.  Nevertheless, a viable alternative, which also satisfies the same worst case risk bound, and which is equally straightforward to compute, is to let $\tilde{h}_n$ denote the (ordinary) log-concave MLE based on $Z_1,\ldots,Z_n$, and then set
\begin{equation}
\label{Eq:fntilde}
\tilde{f}_n(x) := \left\{ \begin{array}{ll} \tilde{h}_n(\|x\|)/(c_d\|x\|^{d-1}) & \mbox{if $x \neq 0$} \\
0 & \mbox{if $x=0$,} \end{array} \right.
\end{equation}
where $c_d := 2\pi^{d/2}/\Gamma(d/2)$.  This estimator, however, ignores the fact that the density of $Z_1$ is a `special' log-concave density, belonging to the class 
\[
\mathcal{H} := \biggl\{r \mapsto r^{d-1}e^{\phi(r)}: \phi \in \Phi^{\mathrm{SS}}, \int_0^\infty r^{d-1}e^{\phi(r)} \, dr = 1\biggr\},
\]
and means that $\tilde{f}_n$ does not belong to $\mathcal{F}_d^{\mathrm{SS}}$ in general.  Moreover, $\tilde{f}_n$ is inconsistent at $x=0$ (the estimator is zero for $\|x\| < \min_i Z_i$) and behaves badly for small $\|x\|$; cf.~Figure~\ref{Fig:SS}, taken from \citet{XuSamworth2017}.
\begin{figure}[ht!]
  \centering
    \includegraphics[scale=0.25, trim=140 0 0 0]{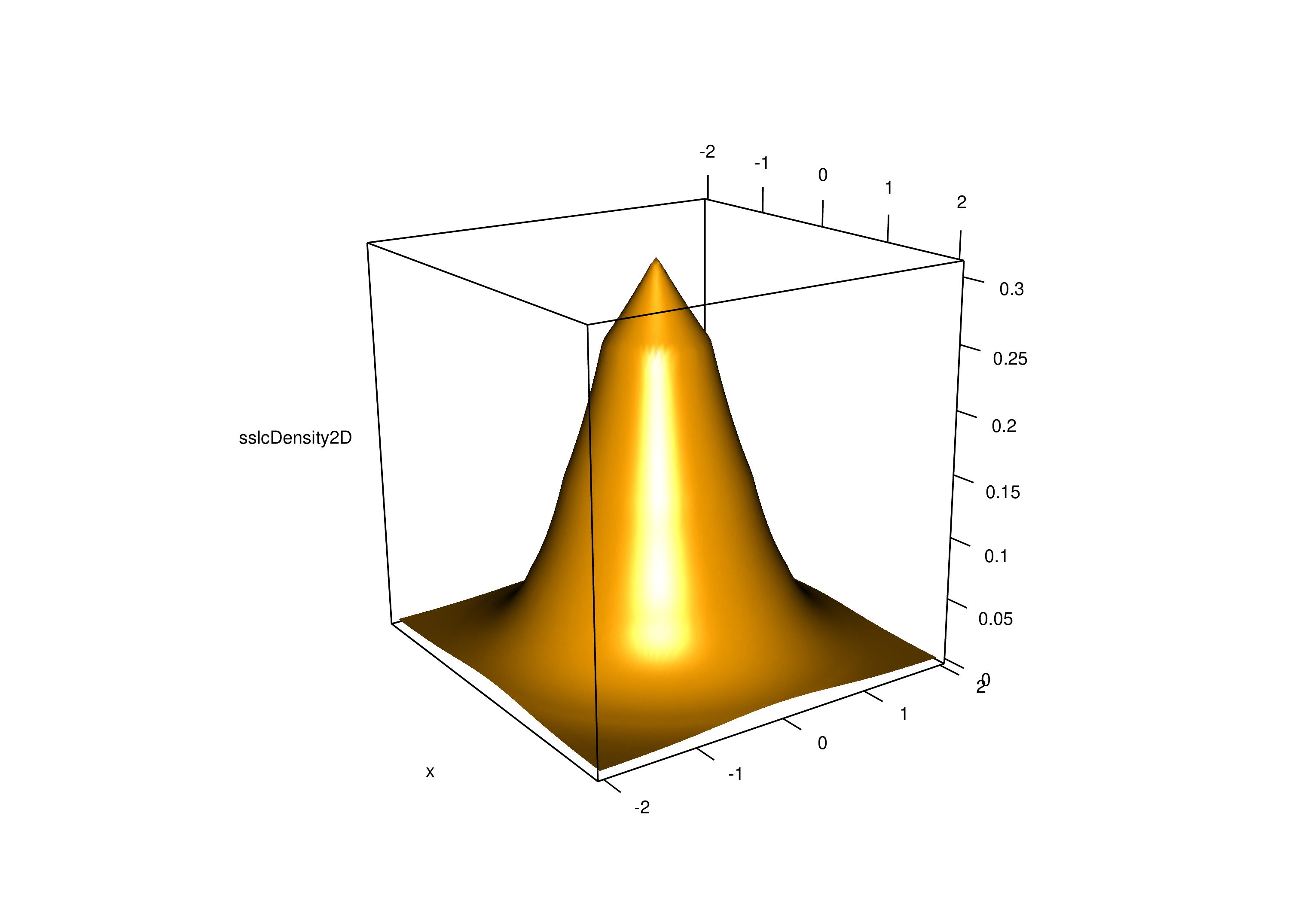}
    \includegraphics[scale=0.25, trim=140 0 0 0]{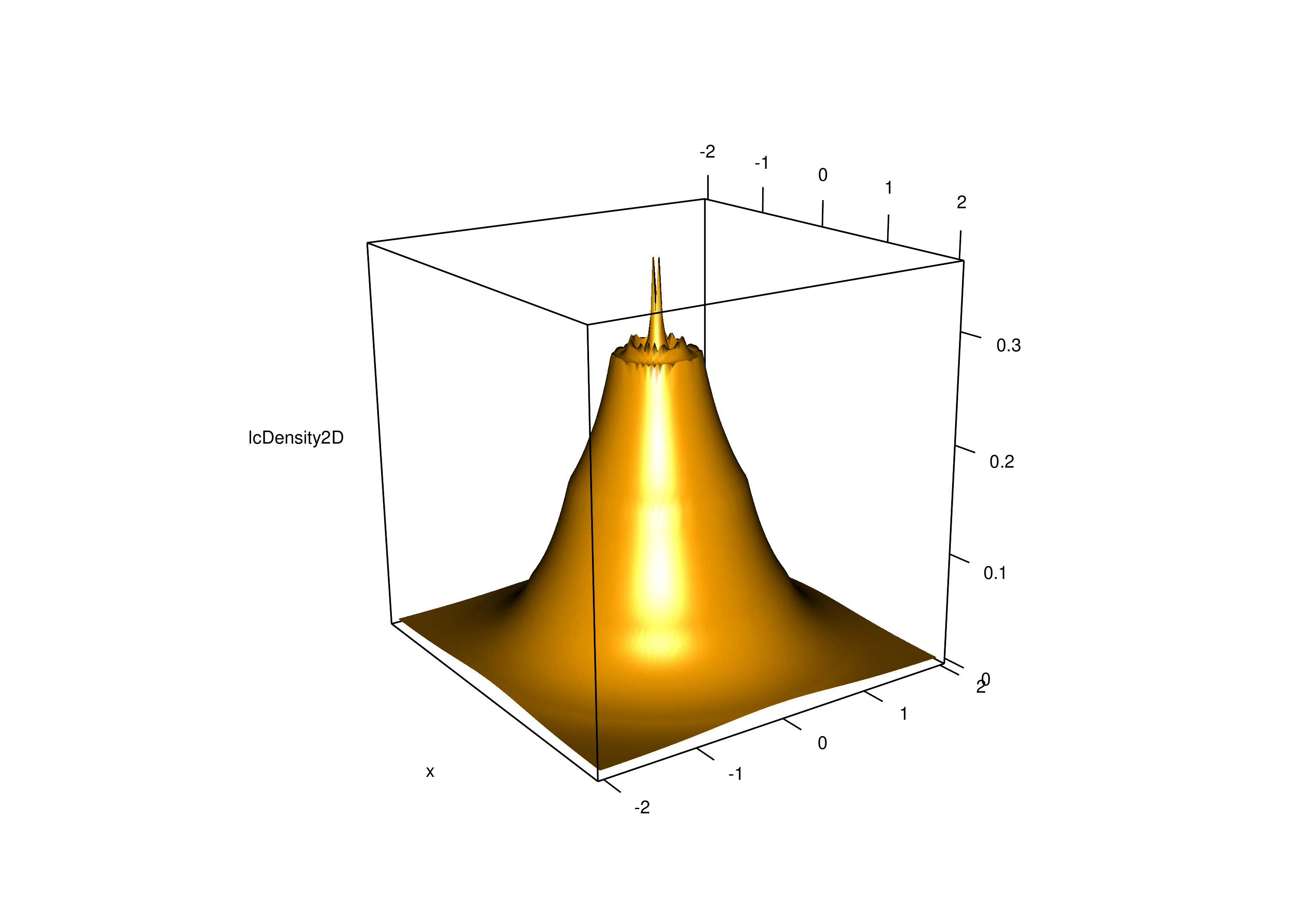}
\caption{\label{Fig:SS}A comparison of the spherically-symmetric log-concave MLE $\hat{f}_n^{\mathrm{SS}}$ (left) and the estimator $\tilde{f}_n$ defined in~\eqref{Eq:fntilde} (right) based on a sample of size $n=1000$ from a standard bivariate normal distribution.}
\end{figure}

A further advantage of $\hat{f}_n^{\mathrm{SS}}$ in this context relates to its adaptation behaviour.  To describe this, for $k \in \mathbb{N}$, we say $\phi \in \Phi^{\mathrm{SS}}$ is \emph{$k$-affine}, and write $\phi \in \Phi^{\mathrm{SS},k}$, if there exist $r_0 \in (0,\infty]$ and a partition $I_1,\ldots,I_{k}$ of $[0,r_0)$ into intervals such that $\phi$ is affine on each $I_j$ for $j=1,\ldots,k$, and $\phi(r) = -\infty$ for $r > r_0$.  Define $\mathcal{H}^{k} := \bigl\{ h \in \mathcal{H}\,:\, h(r) = r^{d-1}e^{\phi(r)} \textrm{ for some } \phi \in \Phi^{\mathrm{SS},k}\bigr\}$.  
\begin{thm}[\citet{XuSamworth2017}]
\label{Thm:Adaptation}
Let $f_0 \in \mathcal{F}_d^{\mathrm{SS}}$ be given by $f_0(x) = e^{\phi_0(\|x\|)}$, where $\phi_0 \in \Phi^{\mathrm{SS}}$ and let $X_1,\ldots, X_n \stackrel{\mathrm{iid}}{\sim} f_0$. Let $\hat{f}_n^{\mathrm{SS}}$ be the spherically symmetric log-concave MLE.  Define $h_0 \in \mathcal{H}$ by $h_0(r) := r^{d-1}e^{\phi_0(r)}$ for $r \in [0,\infty)$.  Then, writing $\nu_k^2 := 2 \wedge \inf_{h \in \mathcal{H}^{k}} d_{\mathrm{KL}}^2(h_0,h)$, there exists a universal constant $C > 0$ such that
  \[
    \mathbb{E}d_X^2(\hat{f}_n^{\mathrm{SS}},f_0) \leq C \min_{k =1,\ldots,n} \biggl(\frac{k^{4/5}\nu_k^{2/5}}{n^{4/5}} \log \frac{en}{k\nu_k} + \frac{k}{n}\log^{5/4}\frac{en}{k}\biggr).
  \]
\end{thm}
Interestingly, this result implies the following sharp oracle inequality: there exists a universal constant $C > 0$ such that
  \[
    \mathbb{E}d_X^2(\hat{f}_n^{\mathrm{SS}},f_0) \leq \min_{k =1,\ldots,n} \biggl(\nu_k^2 +  C\frac{k}{n}\log^{5/4}\frac{en}{k}\biggr).
  \]

\section{Other topics}

\subsection{$s$-concave densities}

As an attempt to allow heavier tails than are permitted by log-concavity, say a density $f$ is \emph{$s$-concave} with $s < 0$, and write $f \in \mathcal{F}_{d,s}$, if $f = (-\phi)^{1/s}$ for some $\phi \in \Phi$.  Such densities have convex upper level sets, but allow polynomial tails, and satisfy $\mathcal{F}_{d,s_2} \supseteq \mathcal{F}_{d,s_1} \supseteq \mathcal{F}_d$ for $s_2 < s_1 < 0$.  Some, but not all, of the properties of $\mathcal{F}_d$ translate over to these larger classes \citep[e.g.][]{DJ1988}.  Results on the maximum likelihood estimator in the case $d=1$ are recently available \citep{DossWellner2016}, but estimation techniques based on R\'enyi divergences are also attractive here \citep{KoenkerMizera2010,HanWellner2016}. 

\subsection{Finite mixtures of log-concave densities}

Finite mixtures offer another attractive way of generalising the scope of log-concave modelling \citep{ChangWalther2007,EilersBorgdorff2007,CSS2010}.  The main issue concerns identifiability: for instance the mixture distribution $p N_d(-\mu,I) + (1-p)N_d(\mu,I)$ with $p \in (0,1)$ has a log-concave density if and only if $\|\mu\| \leq 1$ \citep{CSS2010}.  However, all is not lost: for instance, consider distribution functions on $\mathbb{R}$ of the form
\[
G(x) := pF(x-\mu_1) + (1-p)F(x-\mu_2),
\]
where $p \in [0,1]$, $\mu_1 \leq \mu_2$ and $F(-x) = 1 - F(x)$, so that the distribution corresponding to $F$ is symmetric about zero.  \citet{HWH2007} proved that if $p \notin \{0,1/2,1\}$ and $\mu_1 < \mu_2$, then $p$, $\mu_1$, $\mu_2$ and $F$ are identifiable.  \citet{BalabdaouiDoss2017} have recently exploited this result to fit a two-component location mixture of a symmetric, log-concave density.  One can imagine this as a model for a population of adult human heights, where the two components correspond to men and women.

\subsection{Regression problems}

Consider the basic regression model
\[
Y = m(x) + \epsilon,
\]
where $x \in \mathbb{R}^d$ is considered fixed for simplicity, $m$ belongs to a class of real-valued functions $\mathcal{M}$ and $\epsilon \sim P$ with $\mathbb{E}(\epsilon) = 0$.  There is a large literature on estimating $m$ under different shape constraints \citep[e.g.][]{vanEeden1958, GJW2001, HanWellner2016b, ChenSamworth2016}.  But log-concavity does not seem to be a natural constraint to impose on a regression function.  On the other hand, it may well represent a sensible model for the distribution of the error vector $\epsilon$.  Given covariates $x_1,\ldots,x_n \in \mathbb{R}^d$ and corresponding independent responses $Y_1,\ldots,Y_n$, \citet{DSS2011,DSS2013} considered estimating $(m,\log \psi^*(P))$ by
\[
(\hat{m},\phi^*) \in \argmax_{(m,\phi) \in \mathcal{M} \times \Phi} \frac{1}{n}\sum_{i=1}^n \phi\bigl(Y_i - m(x_i)\bigr) - \int_{\mathbb{R}^d} e^\phi + 1.
\]
Such a maximiser exists, assuming only that $\mathcal{M}$ is closed under the addition of constant functions, and that $\mathcal{M}(x) := \bigl\{\bigl(m(x_1),\ldots,m(x_n)\bigr):m \in \mathcal{M}\bigr\}$ is a closed subset of $\mathbb{R}^n$.  Under a triangular array scheme, it can be shown that in the case of linear regression with a fixed number of covariates, the estimator of the vector of regression coefficients is consistent \citep[][Corollary~2.2]{DSS2013}, while numerical evidence suggests that the estimator can yield significant improvements over the ordinary least squares estimator in settings where $\epsilon$ has a log-concave, but not Gaussian, density.  Similar to the Independent Component Analysis problem studied in Section~\ref{Sec:ICA}, the optimisation problem is again only bi-concave, though stochastic search algorithms offer a promising approach \citep{DSS2013}. 

\appendix


\end{document}